\documentclass[11pt,letterpaper,pdftex]{article}

\newif\ifsubmission
\submissionfalse

\usepackage{fullpage}

\usepackage[font=small,labelfont=bf]{caption}

\usepackage[dvips]{color}
\usepackage{xcolor}
\definecolor{DarkRed}{RGB}{150,0,0}
\definecolor{DarkGreen}{RGB}{0,150,0}
\definecolor{DarkBlue}{RGB}{0,0,150}
\definecolor{purple}{RGB}{200,0,200}
\usepackage[colorlinks=true,linkcolor=DarkRed,citecolor=DarkBlue,urlcolor=DarkGreen]{hyperref}
\usepackage[hyperpageref]{backref}

\usepackage{tikz}
\usetikzlibrary{shapes,positioning, arrows.meta}

\usepackage{beton}
\usepackage[T1]{fontenc}

\usepackage{mathptmx}
\usepackage{subfig}

\usepackage{bm}

\usepackage{paralist,tabularx}
\usepackage{amsmath,amsfonts,amssymb,amsthm}
\usepackage{latexsym}
\usepackage{enumerate}
\usepackage{mathtools}
\usepackage{wrapfig}

\usepackage{algorithm, algorithmicx, algpseudocode} 
\usepackage{algorithmicx}

\usepackage{algorithm}
\usepackage{algorithmicx}
\usepackage{algpseudocode}
\usepackage[algo2e]{algorithm2e}

\newtheorem{theorem}{Theorem}[section]
\newtheorem{lemma}[theorem]{Lemma}

\newtheorem{corollary}[theorem]{Corollary}

\newcounter{definition}
\newenvironment{definition}[1][]{\refstepcounter{definition}\par\medskip
   \noindent \textbf{Definition~\thedefinition. #1} \rmfamily}{\medskip}	

{\makeatletter
 \gdef\xxxmark{%
   \expandafter\ifx\csname @mpargs\endcsname\relax 
     \expandafter\ifx\csname @captype\endcsname\relax 
     \marginpar{\textcolor{red}{\textbf{xxx}}}
          \else
       \textbf{\textcolor{red}{xxx}} 
    \fi
   \else
     \textbf{\textcolor{red}{xxx}} 
   \fi}
 \gdef\xxx{\@ifnextchar[\xxx@lab\xxx@nolab}
 \long\gdef\xxx@lab[#1]#2{\textbf{[\xxxmark \textcolor{{blue}#2} ---{\sc {\color{blue}#1}}]}}
 \long\gdef\xxx@nolab#1{\textbf{[\xxxmark \textcolor{blue}{#1}]}}
}

\newcommand{\klcs}[1][]{\ifthenelse{\equal{#1}{}}{$k$-LCS}{${#1}$-LCS}}

\newcommand{\kwlcs}[1][]{\ifthenelse{\equal{#1}{}}{$k$-WLCS}{${#1}$-WLCS}}

\newcommand{\kNLstC}[1][]{\ifthenelse{\equal{#1}{}}{$k$-NLstC}{${#1}$-NLstC}}

\newcommand{\kELstC}[1][]{\ifthenelse{\equal{#1}{}}{$k$-ELstC}{${#1}$-ELstC}}

\newcommand{\czkc}[1][]{\ifthenelse{\equal{#1}{}}{FZ$k$C}{FZ${#1}$C}}

\newcommand{\czkch}[1][]{\ifthenelse{\equal{#1}{}}{FZ$k$CH}{FZ${#1}$CH}}

\newcommand{\cfkc}[1][]{\ifthenelse{\equal{#1}{}}{F$\mathfrak{f}k$C}{F$\mathfrak{f}{#1}$C}}

\newcommand{\cfkch}[1][]{\ifthenelse{\equal{#1}{}}{F$\mathfrak{f}k$CH}{F$\mathfrak{f}{#1}$CH}}

\newcommand{\ckov}[1][]{\ifthenelse{\equal{#1}{}}{F$k$-OV}{F${#1}$-OV}}

\newcommand{\ckovh}[1][]{\ifthenelse{\equal{#1}{}}{F$k$-OVH}{F${#1}$-OVH}}

\newcommand{\cksum}[1][]{\ifthenelse{\equal{#1}{}}{F$k$-SUM}{F${#1}$-SUM}}

\newcommand{\cksumh}[1][]{\ifthenelse{\equal{#1}{}}{F$k$-SUMH}{F${#1}$-SUMH}}

\newcommand{\fzkc}[1][]{\ifthenelse{\equal{#1}{}}{FZ$k$C}{FZ${#1}$C}}

\newcommand{\ckxor}[1][]{\ifthenelse{\equal{#1}{}}{F$k$-XOR}{F${#1}$-XOR}}

\newcommand{\ckxorh}[1][]{\ifthenelse{\equal{#1}{}}{F$k$-XORH}{F${#1}$-XORH}}

\newcommand{\ckfunc}[1][]{\ifthenelse{\equal{#1}{}}{F$k$-$\mathfrak{f}$}{F${#1}$-$\mathfrak{f}$}}

\newcommand{\ckfunch}[1][]{\ifthenelse{\equal{#1}{}}{F$k$-$\mathfrak{f}$H}{F${#1}$-$\mathfrak{f}$H}}

\newcommand{\erdosRen}{Erd{\H{o}}s-R{\'{e}}nyi }

\newcommand{\kxor}[1][]{\ifthenelse{\equal{#1}{}}{$k$-XOR}{${#1}$-XOR}}

\newcommand{\ksum}[1][]{\ifthenelse{\equal{#1}{}}{$k$-SUM}{${#1}$-SUM}}

\newcommand{\vksum}[1][]{\ifthenelse{\equal{#1}{}}{V$k$-SUM}{V${#1}$-SUM}}

\newcommand{\ukov}[1][]{\ifthenelse{\equal{#1}{}}{U$k$-OV}{U${#1}$-OV}}
\newcommand{\nukov}[1][]{\ifthenelse{\equal{#1}{}}{NU$k$-OV}{NU${#1}$-OV}}
\newcommand{\ukxor}[1][]{\ifthenelse{\equal{#1}{}}{U$k$-XOR}{U${#1}$-XOR}}

\newcommand{\GoodDPolys}[1][]{\ifthenelse{\equal{#1}{}}{Good $d$-Degree Polynomials}{Good ${#1}$-Degree Polynomials}}
\newcommand{\goodDPoly}[1][]{\ifthenelse{\equal{#1}{}}{good $d$-degree polynomial}{good ${#1}$-degree polynomial}}

\newcommand{\OKPolys}[1][]{\ifthenelse{\equal{#1}{}}{Fine $d$-Degree Polynomials}{Fine ${#1}$-Degree Polynomials}}
\newcommand{\okPoly}[1][]{\ifthenelse{\equal{#1}{}}{fine $d$-degree polynomial}{fine ${#1}$-degree polynomial}}

\def\edgepartitionc{E_{a_1,...,a_c}}
\def\gchosenlabels{G^{\ell_1,...\ell_{\binom{k}{c}}}}
\def\Otil{\tilde{O}}
\usepackage{thm-restate}

\title{Worst-Case and Average-Case Hardness of\\
Hypercycle and Database Problems}

\author{
    Cheng-Hao Fu \thanks{Boston University, \texttt{chenghao@bu.edu}.}
    \and 
    Andrea Lincoln \thanks{Boston University, \texttt{andrea2@bu.edu}.}
    \and 
    Rene Reyes \thanks{Boston University, \texttt{rdreyes@bu.edu}. Supported by NSF Grant No. 2209194.}
}
\date{April 25, 2025}
\begin{document}

\maketitle
\thispagestyle{empty}
\begin{abstract}
In this paper we present tight lower-bounds and new upper-bounds for hypergraph and database problems. 
We give tight lower-bounds for finding minimum hypercycles. 
We give tight lower-bounds for a substantial regime of unweighted hypercycle. We also give a new faster algorithm for longer unweighted hypercycles. 
We give a worst-case to average-case reduction from detecting a subgraph of a hypergraph in the worst-case to counting subgraphs of hypergraphs in the average-case. 
We demonstrate two applications of this worst-case to average-case reduction, which result in average-case lower bounds for counting hypercycles in random hypergraphs and queries in average-case databases. Our tight upper and lower bounds for hypercycle detection in the worst-case have immediate implications for the average-case via our worst-case to average-case reductions.

\end{abstract}

\thispagestyle{empty}
\newpage

\ifsubmission
\else
\pagenumbering{roman}
\tableofcontents
\newpage
\pagenumbering{arabic}
\fi 

\section{Introduction}
\label{sec:intro}
The fine-grained hardness of hypergraph problems is known to have interesting connections to both theoretical and practical problems in computer science. For example, hypergraph problems have deep connections to solving constraint satisfaction problems: algorithms for hypercycle problems in hypergraphs can be used to solve various Constraint Satisfaction Problems (CSPs), including Max-$k$-SAT \cite{LVW18}.
On the practical side, a recent line of work in database theory (initiated by Carmeli and Kr\"oll \cite{kroll2021}) has shown that certain queries of interest can be interpreted as problems about detecting, listing and counting small structures in hypergraphs.
Nevertheless, the fine-grained study of hypergraph problems remains largely under-explored, especially when compared to the volume of work on fine-grained hardness of graph problems.
In this paper, we aim to extend the understanding of hypergraphs in the average-case setting.
We focus on two main directions: the worst-case hardness of hypercycle problems and the average-case hardness of counting small subhypergraphs.

To our knowledge, this work is the first to show that the conditional hardness of hypercycle detection is dependent on both the length of the hypercycle and the size of the hyperedges.
This is somewhat surprising, given that for other problems, such as hyperclique detection, hardness is conjectured to depend only on the size of the hyperclique.
In the weighted case we get tight lower bounds for all hypercycle lengths.
In the unweighted case, we show that brute-force is necessary for  short hypercycles, resulting in tight upper and lower bounds. 
For the parameter regime where we can't show the brute force lower bound, we give a new faster algorithm for longer unweighted hypercycles using matrix multiplication. 
This shows the brute-force lower bound is actually false past the point where our reduction stops working.

For our average-case results, we build on a line of work that has used the error-correcting properties of polynomials to show average-case fine-grained hardness.
This technique, introduced by Ball et al. \cite{BallWorstToAvg} has already been applied to hypergraph problems by Boix-Adser\`a et al. \cite{UniformCliqueABB}, who show that counting small hypercliques in random hypergraphs where all hyperedges have the same size is as hard as in the worst-case.
In the graph setting, their approach was generalized by Dalirroyfard et al.\cite{factoredProblems} who showed similar worst-case-to-average-case reductions for counting any small subgraph.
We further generalize their result to show that this is true for counting small subhypergraphs in random hypergraphs, including ones where the hyperedges have different sizes.
The worst-case to average-case reduction also allows us to prove that counting queries on random databases are hard on average.

These two results are highly complementary.
Applying our worst-case-to-average-case reduction to our hypercycle hardness results immediately gives us average-case hardness of hypercycle problems.
Furthermore, understanding the worst-case hardness of other hypergraph substructures would shed light on what types of counting queries are hard on average.
Finally, in the process of exploring these areas, we have found many new avenues to explore, which range from understanding the parameter regimes where we do not get tight upper and lower bounds to improving some generic techniques from traditional worst-case reductions.

\subsection{Hypercycle Algorithms and Lower Bounds}

A $k$-hypercycle in a $u$-uniform hypergraph is a list of $k$ distinct nodes $v_1, \ldots, v_k$ such that every hyperedge of $u$ adjacent nodes exists in the hypergraph. 
That is, a $k$-hypercycle consists of the edges $(v_i, \ldots, v_{i+u-1 \mod k})$ for all $i \in [k]$.
In the literature, this is often called a tight hypercycle, but we will omit the `tight' and call these hypercycles throughout this paper (following the norm of \cite{LVW18}). 

We demonstrate new algorithms and conditional lower bounds for the minimum hypercycle and unweighted hypercycle problems.
Our upper and lower bounds for minimum hypercycle are tight, as shown in Figure \ref{fig:weightedHypercycleBounds}.
In the unweighted case, we show that the hardness of hypercycle detection depends on the relationship between the hyperedge size and hypercycle length, as shown in Figure \ref{fig:unweightedHypercycleBounds}.
The conditional lower bounds are derived from the minimum $k$-clique hypothesis and the $(k, u)$-hyperclique hypothesis:

\begin{definition}[MIN WEIGHT k-CLIQUE HYPOTHESIS (e.g. \cite{AbboudWW14,BackursT17})] 
\label{def:minClique}There is a constant $c > 1$ such that, on a Word-RAM with $O(\log(n))$-bit words, finding a $k$-Clique of minimum total edge weight in an $n$-node graph with non-negative integer edge weights in
$[1,n^{ck}]$ requires $n^{k-o(1)}$ time.
\end{definition}

\begin{definition}[$(k, u)$-HYPERCLIQUE HYPOTHESIS \cite{LVW18}]
\label{def:hypercliqueHypothesis}
Let $k > u > 2$ be integers. On a Word-RAM with $O(\log(n))$ bit
words, finding a $k$-hyperclique in a 
$u$-uniform hypergraph on $n$ 
nodes requires $n^{k-o(1)}$ time.
\end{definition}

As evidence for this hypothesis, \cite{LVW18} give a reduction from the maximum degree-$u$-CSP problem to the $(k,u)$-hyperclique problem. 
They also give a reduction from hyperclique detection to hypercycle detection.
The hard instances output by this reduction only cover a specific hypercycle length $k$ for each uniformity $u$ and this seems inherent to the technique.
A natural question to ask is whether the hardness of finding $k$-hypercycles in $u$-uniform hypergraphs depends on the relationship between $u$ and $k$.
We answer this question in the affirmative for both weighted and unweighted hypercycle problems.
This is surprising since the hardness of other hypergraph problems such as $(k,u)$-hyperclique is conjectured to be independent of this relationship.

For min-weight $k$-hypercycle, we show that brute-force is required when $k$ is less than $2u-1$.
When $k \geq 2u-1$, we give an algorithm with runtime independent of $k$.
For unweighted hypercycle, we show that brute force is needed for all $k$ up to the length of the instances output by the \cite{LVW18} reduction.
For longer hypercycles, we show this is not the case by giving novel algorithms which use matrix multiplication to get a speedup.

We now provide a more detailed overview of our techniques along with formal theorem statements for our upper and lower bounds.



\subsubsection{Weighted Hypercycle}
In Section \ref{sec:tight_weighted} we give the results for weighted hypercycle, depicted in Figure \ref{fig:weightedHypercycleBounds}. 
The matching lower bounds come from a reduction between cliques in graphs (2-uniform hypergraphs) and hypercycles. Informally the minimum $k$-clique hypothesis states that finding a minimum $k$-clique in a graph can't be done faster than the naive algorithm which takes $n^{k-o(1)}$ time (see definition \ref{def:minClique}). We prove the following. 
\begin{restatable}{theorem}{weightedCycleLB}
    If the minimum $k$-clique hypothesis holds then the minimum $k$-hypercycle problem in a $u$-uniform hypergraph requires $n^{k-o(1)}$ time for $k \in [u+1 , 2u-1]$. 
\end{restatable}

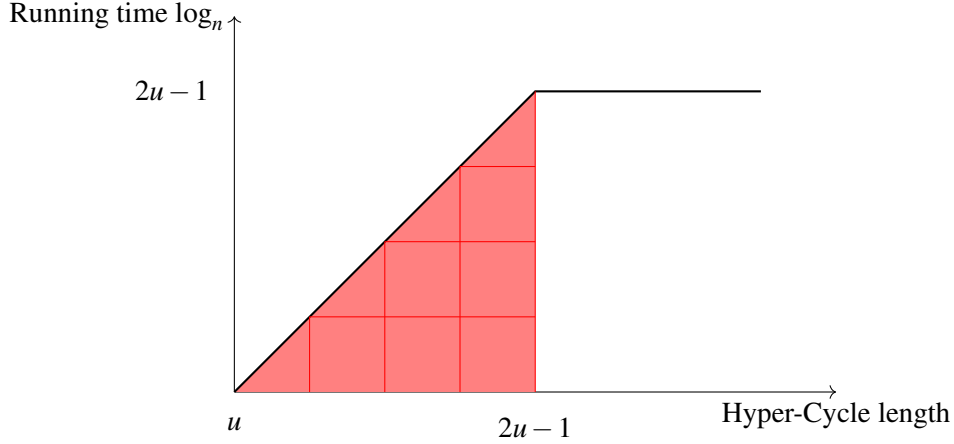
\begin{figure} [ht]
    \centering
    \begin{tikzpicture}
    \definecolor{shadecolor}{rgb}{1,0.5,0.5}

    \draw[->] (0,0) -- (8,0) node[anchor=north] {Hyper-Cycle length};
    \draw[->] (0,0) -- (0,5) node[anchor=east] {Running time $\log_n$};

    \fill[shadecolor] (0,0) -- (4,4) -- (4,0) -- cycle;

    \draw[thick] (0,0) -- (4,4) -- (7,4);

    \draw[red, thin] (1,0) -- (1,1);
    \draw[red, thin] (2,0) -- (2,2);
    \draw[red, thin] (3,0) -- (3,3);
    \draw[red, thin] (4,0) -- (4,4);
    \draw[red, thin] (1,1) -- (4,1);
    \draw[red, thin] (2,2) -- (4,2);
    \draw[red, thin] (3,3) -- (4,3);

    \node[anchor=north] at (0,-0.2) {$u$};
    \node[anchor=north] at (4,-0.2) {$2u-1$};
    \node[anchor=east] at (-0.2,4) {$2u-1$};

\end{tikzpicture}
    \caption{The running time of minimum weight hyper-cycle in a weighted $u$-uniform hyper-graph. The exponent of the running time is indicated by the black line and the lower bound is matching, as indicated by the hatched red area. Note the running time is $O(n^u)$ for a $u$ length hyper-cycle and $O(n^{2u-1})$ for a $2u-1$ length hyper-cycle. Then for all hyper-cycles of length $k>2u-1$ the running time continues to be $O(n^{2u-1})$.}
    \label{fig:weightedHypercycleBounds}
\end{figure}

\subsubsection{Unweighted Hypercycle} 

For the unweighted version of the problem we get tight upper and lower bounds in a $u$-uniform graph for all $k \leq \gamma_3^{-1}(u)$, where we define $ \gamma_3(k) = k-\lceil k/3 \rceil +1$ as is done in \cite{LVW18} and $k = \gamma_3^{-1}(u)$ is the \emph{largest} $k$ such that $\gamma_3(k)=u$.
Intuitively, this is the largest $k$ such that any three nodes in the hypercycle are covered by at least one hyperedge of size $u$, which happens when $k \approx 3u/2$. 
For larger hypercycles where $k \geq \gamma_3^{-1}(u)+1$, we can pick three partitions such that no hyperedge covers all three. 
This property allows us to reduce the problem to triangle-counting and then use matrix multiplication to obtain an algorithmic speedup.
This algorithm, given in Section \ref{sec:ub_unweighted}, runs in time $n^{k-3+\omega}$ where $\omega$ is the matrix multiplication constant. 
We depict the upper and lower bounds for unweighted hypergraphs in Figure \ref{fig:unweightedHypercycleBounds}.

We get tight lower bounds for hypercycle up to the ``phase transition'' point at $ \gamma_3(k)$. These lower bounds come from the $(u,k)$-hyperclique hypothesis. Informally, this says detecting a $k$-hyperclique in a $u$-uniform hypergraph if $k > u > 2$ requires $n^{k-o(1)}$ (see Definition \ref{def:hypercliqueHypothesis}).  We prove the following theorem in section \ref{subsec:tighthardnessshorthypercycles}.

\begin{restatable}{theorem}{tightShortHypercycle}\label{thm:tight short hypercycle}
    Under the $(3,k)$-hyperclique hypothesis, an algorithm for finding (counting)  $(u,k)$-hypercycle for $k \in [u,\gamma_3^{-1}(u)]$ requires $O(n^{k-o(1)})$ time.
\end{restatable}

When $k > \gamma_3(k)$ we show an algorithm which is better than brute force by a factor dependent on $\omega$ (where $\omega$ is the matrix multiplication constant). 
\begin{restatable}{theorem}{fastCycleDetect} \label{thm:fastCycledetect}
    There exists a time-$\tilde{O}(k^k(n+m+n^{2u-1-(3-\omega)}))$ algorithm for finding $k$-hypercycles in any $n$-node hypergraph $G$ when $k \geq 2u-1$. 
\end{restatable}

Proving conditional lower bounds in this regime might unexpectedly imply conditional lower bounds on $\omega$.
In short, we get tight lower bounds for all values of $k$ where lower bounds wouldn't have implications for matrix multiplication's running time, $\omega$.

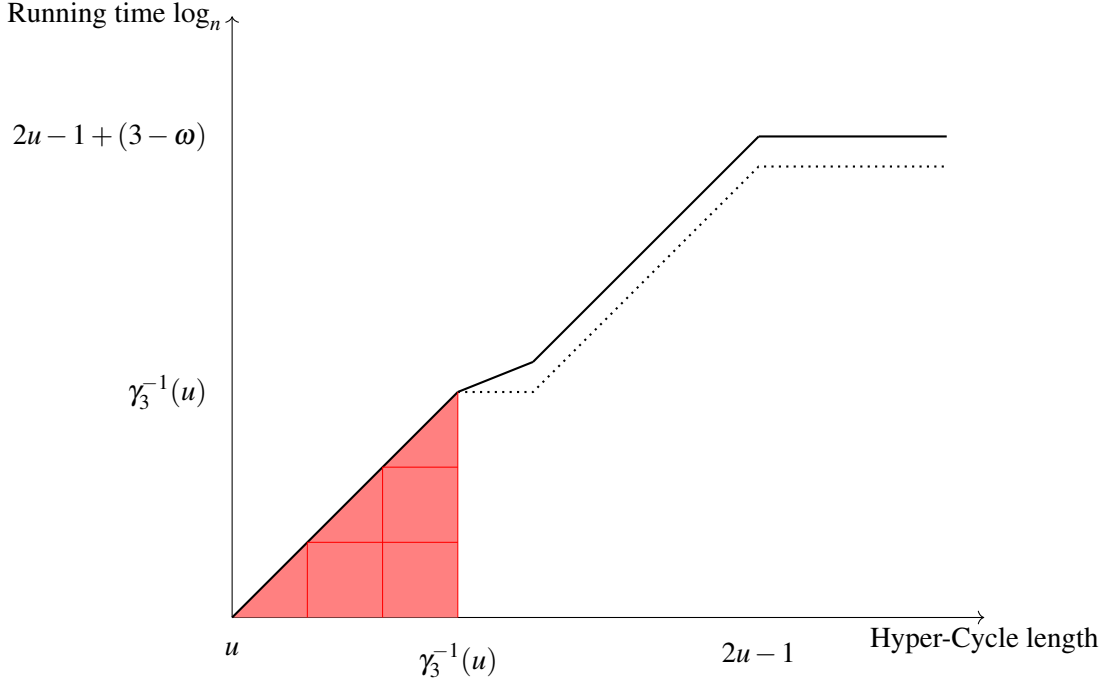
\begin{figure} [ht]
    \centering
        \begin{tikzpicture}
    \definecolor{shadecolor}{rgb}{1,0.5,0.5}
    \definecolor{linecolor}{rgb}{1,0,1}

    \draw[->] (0,0) -- (10,0) node[anchor=north] {Hyper-Cycle length};
    \draw[->] (0,0) -- (0,8) node[anchor=east] {Running time $\log_n$};

    \fill[shadecolor] (0,0) -- (3,3) -- (3,0) -- cycle;

    \draw[thick] (0,0) -- (3,3);
    \draw[thick] (3,3) -- (4,3.4);
    \draw[thick] (4,3.4) -- (7,6.4);
    \draw[thick] (7,6.4) -- (9.5, 6.4);

    \draw[dotted, thick] (3,3) -- (4,3) -- (7,6) -- (9.5,6);

    \draw[red, thin] (1,0) -- (1,1);
    \draw[red, thin] (2,0) -- (2,2);
    \draw[red, thin] (3,0) -- (3,3);
    \draw[red, thin] (1,1) -- (3,1);
    \draw[red, thin] (2,2) -- (3,2);

    \node[anchor=north] at (0,-0.2) {$u$};
    \node[anchor=north] at (3,-0.2) {$\gamma_3^{-1}(u)$};
    \node[anchor=north] at (7,-0.2) {$2u-1$};

    \node[anchor=east] at (-0.2, 3) {$\gamma_3^{-1}(u)$};
    \node[anchor=east] at (-0.2,6.4) {$2u-1+(3-\omega)$};

    \end{tikzpicture}
    \caption{The running time of \textit{unweighted} hyper-cycle in a $u$-uniform hyper-graph. The exponent of the running time is indicated by the black line. The lower bound is matching for the hatched red area from cycles of length $u$ to $\gamma_3^{-1}(u)$ and the running time is $O(n^k)$ for a $k$-hypercycle. Then from $\gamma_3^{-1}(u)$ to $2u-1$ the running time is $\tilde{O}(n^{k-3+\omega})$ for $k$-hypercycle where $\omega$ is the matrix multiplication constant ($\omega < 2.3716$ \cite{mmConstantOmega}). The dotted line represents the running time the algorithm would achieve if $\omega=2$. There is an algorithm running in $\tilde{O}(n^{2u-1-(3-\omega)})$ time for all $k \geq 2u-1$. The shading stops after $\gamma_3^{-1}(u)$ because we don't know of any tight lower bounds past this point. }
    \label{fig:unweightedHypercycleBounds}
\end{figure}

\subsubsection{Average-Case Implications}
We can apply our worst-case to average-case reduction techniques to get lower bounds for random hypergraphs. 
Notably, for constant-length unweighted-hypercycle, the hardness of counting in the worst-case is equivalent to the hardness in the average-case up to polylog factors. 
This gives tight upper and lower bounds for the average-case hardness of counting $k$-hypercycles in the same regime of cycle length for a given uniformity when $k \leq \gamma_3^{-1}(u)$. Informally, if the $(3,k)$-hyperclique hypothesis holds then counting $k$-hypercycles in $u$-uniform Erd{\H{o}}s-R{\'{e}}nyi hypergraphs when $k \in [u, \gamma_3^{-1}(u)]$ requires $\tilde{O}(n^{k-o(1)})$ time. We prove this in Section \ref{sec:lb_unweighted}.

\subsection{Counting Subhypergraphs on Average is as Hard as Detecting them in the Worst Case }
To enable our average-case results for databases and for counting hypercycles we provide a new reduction. 
Our reduction can handle non-uniform hypergraphs, meaning hypergraphs with edges of \emph{different} sizes (see Figure \ref{fig:exampleOfSmallGraph}).
This is necessary for the results to apply to the databases setting.
We transform the problem into a low degree polynomial and use a known framework to show average-case hardness. 

However, to get back to average-case graphs that aren't color-coded we use a new form of inclusion exclusion. 
Prior work introduced inclusion-edgesculsion \cite{LVW18}. 
We present a simultaneously generalized and simplified proof which allows us to show hardness for counting these subhypergraphs in uniformly randomly sampled hypergraphs (even those where there are mixed hyperedge sizes). 

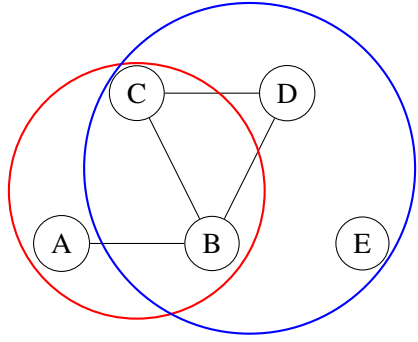
\begin{figure}[ht]
    \centering
    \begin{tikzpicture}

\node[draw, circle] (A) at (0,0) {A};
\node[draw, circle] (B) at (2,0) {B};
\node[draw, circle] (C) at (1,2) {C};
\node[draw, circle] (D) at (3,2) {D};
\node[draw, circle] (E) at (4,0) {E};

\draw (A) -- (B);
\draw (B) -- (C);
\draw (C) -- (D);
\draw (B) -- (D);

\draw[red, thick] (1,0.7) circle [radius=1.7];

\draw[blue, thick] (2.5, 1) circle [radius=2.2];

\end{tikzpicture}
    \caption{An example of a small subgraph with hyperedges of multiple different sizes. This graph contains hyperedges $\{A,B\}, \{B,C\}, \{B, D\}, \{C, D\}, \{A,B,C\},$ and $\{B,C,D,E\}$. We depict the $3$-width edge as a red circle, the $4$-width edge as a blue circle, and the $2$-width edges as black lines.}
    \label{fig:exampleOfSmallGraph} 
\end{figure}

\begin{theorem} [Informal]\label{thm: informal wc to ac}
Let $k$ be a constant. 
Counting the number of $k$-node hypergraphs $H$ made up of nodes $v_1,\ldots, v_k$ in $k$-partite hypergraphs $G$ with partitions $V_1, \ldots, V_k$ in the worst case where we only count copies of $H$ where $v_i \in V_i$ can be solved with polylog calls to a counter for hypergraphs $H$ in uniformly random hypergraphs.
\end{theorem}

Via color-coding this means that if detecting a hypergraph $H$ in the worst case is $T(n)$ hard then counting that hypergraph in the uniform average-case is hard. So, starting from the hardness of either detection or the counting problem in a $k$-partite graph implies the hardness of the counting problem in the uniform average-case. 

These approaches allow us to show hardness for database problems and the problem of counting subhypergraphs. We demonstrate the usefulness of the improved approach by showing its applications to these two problem areas.

\subsection{Database Results}
We now provide a high-level introduction to the databases setting we are interested in. For full definitions of these problems, see Section \ref{sec:database}. 
A table in a database is called a relation. 
A relation, $R_i$, is a set of tuples. 
For example, a relation on $(user\_id, user\_name)$ would be a set of tuples grouping user ids and user names. 
A database may have many relations, and the relations may have tuples of different sizes (e.g. the example database could also include a table of $(user\_id, purchase\_id, purchase\_date)$). 

Queries over a database can take many forms, but a classic form is a join query. 
To be concrete, if we have a table with two relations $R_1(a,b)$ and $R_2(a,c,d)$ then we can have a query $Q_1(a,b,c,d) \leftarrow R_1(a,b), R_2(a,c,d)$. 
If we ask to enumerate all answers to $Q_1$ then we want all tuples $(a',b',c',d')$ such that $(a',b') \in R_1$ and $(a',c',d')\in R_2$. 
Some query types (e.g. `FIRST' in SQL) ask to give a single answer. 
In this paper we will focus on queries that ask to count the number of matching outputs in the query (e.g. `COUNT' in SQL). 
In database theory, enumeration queries are the most studied.
While these queries are often trivial in the average-case, non-enumerating query types are still often hard-on-average.
See Appendix \ref{subsec:whyEnumIsEasy} for more discussion of enumeration and counting on average. 

You might notice that relations look like lists of hyperedges and queries look like requests to enumerate, detect, or count subhypergraphs of the implied hypergraph. 
However, these database ``hyperedges'' are directed. 
The tuple $('Jane', 'Doe')$ and $('Doe', 'Jane')$ have different meanings, whereas in an undirected hypergraph the edges are simply sets. 
However, these problems are similar enough that we can apply analogous worst-case to average-case reduction techniques to show hardness on average for these database problems. 

We now work through an example to help solidify the similarity between databases and hypergraphs.
The hypergraph from Figure \ref{fig:exampleOfSmallGraph} could be represented as a database by the following list of relations: 
$$R_1(A,B), R_2(B,C), R_3(B, D), R_4(C, D), R_5(A,B,C), R_6(B,C,D,E).$$
Then, to count the number of appearances of the hypergraph from Figure \ref{fig:exampleOfSmallGraph} within the hypergraph representing the entire database, we could make the following query:
$$Q_1(A,B,C,D,E) \leftarrow R_1(A,B), R_2(B,C), R_3(B, D), R_4(C, D), R_5(A,B,C), R_6(B,C,D,E).$$ 

A natural question for databases is this: how hard are queries on uniformly random databases? 
If our real world case involves uniformly distributed data, when can we hope to find algorithmic improvements?
We answer this question by giving lower bounds for counting queries in Section \ref{sec:database} and explaining why \emph{small} enumeration queries will be easy in Appendix \ref{sec:appendix_discussion}. 

\subsection{Context and Prior Works}

There has been a lot of work on average-case fine-grained complexity. Ball et al. 
\cite{BallWorstToAvg} showed that starting from popular fine-grained hypotheses, evaluating certain functions could be shown to be hard on average. Later Ball et al. \cite{BallRSV18} showed that you can build a fine-grained proof of work using these functions. 
\cite{UniformCliqueABB} showed that there is an equivalence, up to polylog factors, between counting constant sized hypercliques in the worst-case and in the average-case of  Erd{\H{o}}s-R{\'{e}}nyi Hypergraphs. 
Goldreich presented an alternative method to count cliques mod 2 \cite{Goldreich20}. 
Dalirrooyfard et al. \cite{factoredProblems} generalized these ideas and presented a method to produce a worst-case to average-case reduction for any problem which can be represented by a `good' low degree polynomial.
Additionally they showed that counting small subgraphs (not just cliques) is equivalently hard in the worst-case and in Erd{\H{o}}s-R{\'{e}}nyi graphs. 
In this paper we further generalize this result to the hypergraph setting.
Specifically, we show that counting small sub\textbf{hyper}graphs is equivalently hard in the worst-case and average-case. 
We also give worst-case to average-case hardness for many classes of \textbf{counting database queries}. 
Our contributions include expanding the previous results and connecting fine-grained average-case complexity to database theory.

Recently, the database theory community has had increased interest in the fine-grained hardness of various database queries. 
This is highlighted by many recent papers in the area. 
Carmeli and Kröll \cite{kroll2021} use fine-grained hypotheses to show that answering unions of conjunctive queries are hard.
Carmeli and Segoufin \cite{selfJoinsHard2023} explore the fine-grained complexity of enumeration queries with self joins. 
Bringman et. al. \cite{tightFineGrainedCarmeli} characterize the amount of preprocessing time needed to get polylogarithmic access time. 
We will show that for counting queries over a constant number of relations that don't involve self-joins, the worst-case and average-case hardness are equivalent up to polylogarithmic factors.

Carmeli et al \cite{CarmeliTGKR23} eschew enumeration to explore  logarithmic time access to the $k^{th}$ answer to a conjunctive query after a quasi-linear database preprocessing. 
They provide fast algorithms for these problems.
Recent work from Wang, Willsey, and Suciu \cite{WangWS23} proposed a \emph{free join} which achieves worst-case optimality. 
These are algorithms where the runtime matches the worst-case output size. 
As we explore counting queries where the output size is a single word, $O(\lg(n))$ bits, worst-case optimality is impossible as long as the full input must be read to answer the query. 
We thus focus on efficiency more generally. 
A 2023 Simons program on Logic and Algorithms in Database Theory and AI ran a workshop on Fine-Grained Complexity, Logic, and Query Evaluation\cite{simons_fine_grained_complexity_2023}. 
In this workshop, the open problem of the average-case hardness of database queries was brought up in the open problem session on September 27th. 
In this paper we have explored and proved this average-case hardness for many kinds of counting queries in databases. 
In this paper we seek to give a new definition of an average-case for database theory and provide a worst-case to average-case reduction using fine-grained complexity techniques. 

Hypercycles (``tight hypercycles'' in much of the literature) have been well-studied in both math and computer science. 
Allen et. al. \cite{allen2015tight} showed that for $k$-hypercycles in a $u$-uniform graph where $k \equiv 0 \mod u$ must exist if there are at least $(k/n + \delta)\binom{n}{u}$ hyperedges for some constant $\delta>0$ . 
It is unclear if when $\delta = \Theta(1/n)$ you can still guarantee a hypercycle exists. In this paper we show hardness for $k < 2u$, where these results do not apply. 
Huang and Ma \cite{tightCyclesHypergraphsExistance} continued the exploration of this existential hypercycle question, showing that, in some sense, the previous result is tight up to constants. 
Later, Allen et al. \cite{AllenKPP18} gave a faster algorithm for tight Hamiltonian hypercycles in random graphs, that is in average-case graphs. 
The results in our paper rely on $k$ being small, for larger $k$ the worst-case to average-case reduction becomes increasingly expensive. 
So their work and ours shows an interesting dynamic where short cycles have more similar hardness in the worst-case and average-case and there is a divergence as the cycle length grows.
Lincoln, Vassilevska and Williams \cite{LVW18} give a reduction from max-$k$-SAT to hypercycle and then to cycle. 
However, their results for cycle are not tight. 
They give a bound of $m^{3/2-o(1)}$ for $k$-cycle. They give a tighter bound for $7$-cycle of $m^{7/5-o(1)}$ in graphs with $m = n^{5/4}$. 
In theorem C.1 they show a lower bound for hypercycle. 
However, this is \textit{dense} hypercycle. 
In our paper we present\textit{ tight} bounds for sparse hypercycle. 
We also present a new faster algorithm for unweighted hypercycle. 
Our weighted lower bounds are tight. 
Our unweighted lower bounds are tight until the regime where matrix multiplication is used in the fastest algorithm. 

\subsection{Paper Structure}

We present the basic background and definitions in Section \ref{sec:prelims}. 
We show tight lower bounds for hypercycles in Section \ref{sec:lb_unweighted}. 
We give fast algorithms for hypercycles in Section \ref{sec:ub_unweighted}. 
We give the tight upper and lower bounds for minimum hypercycle in Section \ref{sec:tight_weighted}. 
We give the worst-case to average-case reduction for counting subhypergraphs in Section \ref{sec:wc_to_ac}. 
We show the average-case hardness for (self-join-free) counting queries in Section \ref{sec:database}. 
We present open problems in Section \ref{sec:conc_and_open_prob}. 
Finally, we include extra discussion in Appendix \ref{sec:appendix_discussion}.

\section{Preliminaries}
\label{sec:prelims}
In this paper we combine ideas and techniques from average-case fine-grained complexity, databases, and worst-case fine-grained complexity and algorithms. As a result we will mostly define the needed terms in the corresponding sections. Our preliminaries are short and focus on hypergraphs as these appear in most sections.

\begin{definition}
    A $u$-uniform hypergraph $G$ has a vertex set $V$ and a set of hyperedges $E$ where each hyperedge is a set of $u$ vertices from $V$. 

    A hypergraph of mixed sizes with hyperedge set sizes of $u_1, \ldots, u_k$ has a vertex set $V$ and a set of hyper edges $E$ where each hyperedge is a set of vertices from $V$ and the size of that set must be $u_1, \ldots,$ or $ u_k$.
\end{definition}

\begin{definition}
    A tight $k$-hypercycle in a $u$-uniform hypergraph $G$ is a list of $k$ nodes $v_1, \ldots, v_k \in V$ where every hyperedge
    $$(v_i, v_{i+1 \mod k}\ldots, v_{i+u-1 \mod k})$$
    exists in $E$. That is, every hyperedge of $u$ adjacent vertices on the cycle exists.
\end{definition}


\begin{definition}
    A $k$-hypercycle in a $u$-uniform hypergraph $G$ is a list of $k$ nodes $v_1, \ldots, v_k \in V$ where every hyperedge $(v_{i_1},\ldots, v_{i_u})$ where $i_j \in [1,k]$ exists $E$. That is, every possible hyperedge between the $k$ nodes exists in the graph. 
\end{definition}

\begin{figure}[ht]
\centering
\begin{tikzpicture}[scale=1.0, every node/.style={circle,draw,inner sep=1pt}]

\node[draw=none] at (0.5,3.7) {\Large $H$};

\node (t1) at (0,2) {};
\node (t2) at (1,2) {};
\node (t3) at (0,1) {};
\draw (t1)--(t2)--(t3)--cycle;

\node (b1) at (0,0) {};
\node (b2) at (1,0) {};
\node (b3) at (1,-1) {};
\node (b4) at (0,-1) {};
\draw (b1)--(b2)--(b3)--(b4)--cycle;

\node[draw=none] at (5,3.7) {\Large Example of an $H$-partite Graph};

\begin{scope}[xshift=4.5cm, yshift=1.3cm]
  \foreach \i in {1,...,4} {
    \node (A\i) at (-1,1.6-0.5*\i) {};
  }
  \foreach \i in {1,...,4} {
    \node (B\i) at (0,1.6-0.5*\i) {};
  }
  \foreach \i in {1,...,4} {
    \node (C\i) at (1,1.6-0.5*\i) {};
  }

  \foreach \i in {1,...,4} {
    \foreach \j in {1,...,4} {
      \draw (A\i)--(B\j);
      \draw (B\i)--(C\j);
      \draw (A\i)--(C\j);
    }
  }

  \draw[dotted] (-1,0.8) ellipse (0.4cm and 1.2cm);
  \draw[dotted] ( 0,0.8) ellipse (0.4cm and 1.2cm);
  \draw[dotted] ( 1,0.8) ellipse (0.4cm and 1.2cm);
\end{scope}

\begin{scope}[xshift=4.5cm, yshift=-1.8cm]
  \foreach \i in {1,...,4} {
    \node (P\i) at (-1.5,1.6-0.5*\i) {};
  }
  \foreach \i in {1,...,4} {
    \node (Q\i) at (-0.5,1.6-0.5*\i) {};
  }
  \foreach \i in {1,...,4} {
    \node (R\i) at (0.5,1.6-0.5*\i) {};
  }
  \foreach \i in {1,...,4} {
    \node (S\i) at (1.5,1.6-0.5*\i) {};
  }

  \foreach \i in {1,...,4} {
    \foreach \j in {1,...,4} {
      \draw (P\i)--(Q\j);
      \draw (P\i)--(R\j);
      \draw (P\i)--(S\j);
      \draw (Q\i)--(R\j);
      \draw (Q\i)--(S\j);
      \draw (R\i)--(S\j);
    }
  }

  \draw[dotted] (-1.5,0.8) ellipse (0.4cm and 1.2cm);
  \draw[dotted] (-0.5,0.8) ellipse (0.4cm and 1.2cm);
  \draw[dotted] ( 0.5,0.8) ellipse (0.4cm and 1.2cm);
  \draw[dotted] ( 1.5,0.8) ellipse (0.4cm and 1.2cm);
\end{scope}

\end{tikzpicture}

\caption{An example of two small graphs and corresponding $H$-partite graphs. }
\label{fig:Hpartite}
\end{figure}
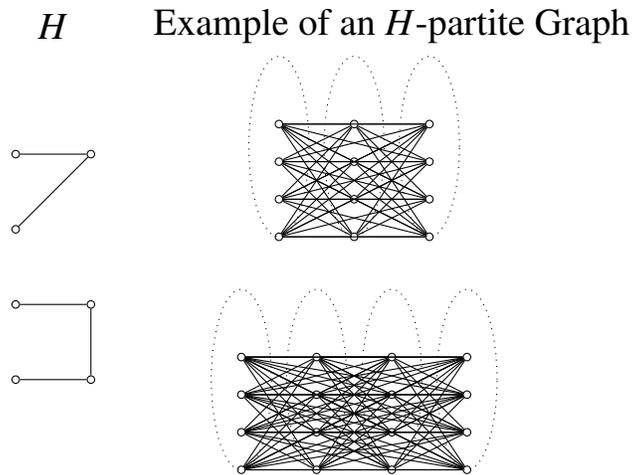

In this paper we use the term $k$-hypercycle to refer to a tight $k$-hypercycle. We use these concepts throughout the paper. We also use a more general notion of a subgraph of a hypergraph. 

\begin{definition}
    A subhypergraph or subgraph of a hypergraph $G$ (the hypergraph could have mixed uniformity or not) is a graph $H$ whose vertex set and edge set is a subset of $G$'s vertex and edge set. That is, if $H$ has a vertex set $V_H$ and edge set $E_H$ then $H$ is a subgraph of $G$ if $V_H \subset V$ and $E_H \subset E$.
\end{definition}

\begin{theorem}[Theorem 3.1 from \cite{LVW18}] \label{thm:hclique_to_hcycle}

    Let $G$ be a $u$-uniform hypergraph on $n$ vertices $V$, partitioned into $k$ parts $V_1,\ldots,V_k$. 
    
    Let $\gamma_u(k)=k-\lceil k/u\rceil +1$.
    
    In $O(n^{\gamma_u(k)})$ time we can create a $\gamma_u(k)$-uniform hypergraph $G'$ on the same node set $V$ as $G$, so that $G'$ contains an $k$-hypercycle if and only if $G$ contains an $k$-hyperclique with one node from each $V_i$.
    
    If $G$ has weights on its hyperedges in the range $[-W,W]$, then one can also assign weights to the hyperedges of $G'$ so that a minimum weight $k$-hypercycle in $G'$ corresponds to a minimum weight $k$-hyperclique in $G$ and every edge in the hyperclique has weight between $[-\binom{\gamma_u(k)}{u}W,\binom{\gamma_u(k)}{u}W]$. Notably, $\binom{\gamma_u(k)}{u}\leq O(k^u)$.
\end{theorem}

For intuition on this theorem see Appendix \ref{subsec:prevCliquetoCycle}.

\begin{definition}[From  \cite{factoredProblems}]
Let $H=(V_H,E_H)$ be a $k$-node graph with $V_H=\{x_1,\ldots,x_k\}$. 

An $H$-partite graph is a graph with $k$ partitions $V_1,\ldots,V_k$. This graph must only have edges between nodes $v_i \in V_i$ and $v_j \in V_j$ if e $(x_i,x_j)\in E_H$. (See Figure \ref{fig:Hpartite} in the appendix.)
\end{definition}








\section{Short Unweighted Hypercycles Require Brute-Force}
\label{sec:lb_unweighted}

We will begin by showing that short hypercycles are tightly hard. We will do so for hypercycles of lengths at most $\gamma_3^{-1}(u)$. Recall that $\gamma_3(k) = k - \lceil k/3 \rceil +1$, so $\gamma_3^{-1}(u) \approx 3u/2$.  On an intuitive level $\gamma_3^{-1}(u)$ corresponds to the largest length of cycle such that all sets of three nodes are included in at least one hyperedge. In this section we will show that given this constraint, the best algorithm for $k$-hypercycle is the naive brute force $n^k$ algorithm. In the next section we will show that if the cycle length is even one larger than $\gamma_3^{-1}(u)$, then we can beat brute force using fast matrix multiplication. In this sense our brute force lower-bound is tight, we couldn't extend the brute force lower bound to cycles of any longer length.

Our lower bounds are based on the min-weight $k$-clique and $(u,k)$-hyperclique hypotheses, which states that these problems require $O(n^{k-o(1))})$ time (see Definition \ref{def:hypercliqueHypothesis}). 
We use a reduction from \cite{LVW18} to show hardness for our hypercycle problems; to make this work self-contained, we now give the necessary intuition.

\subsection{Reduction from k-clique}
\label{subsec:prevCliquetoCycle}

Our lower bounds are based on the min-weight $k$-clique and $(u,k)$-hyperclique hypotheses, which states that these problems require $O(n^{k-o(1))}$ time (see Definition \ref{def:hypercliqueHypothesis}).
We will use a reduction from \cite{LVW18} to show hardness for our hypercycle problems.

\begin{theorem}[Theorem 3.1 from \cite{LVW18}] 

    Let $G$ be a $u$-uniform hypergraph on $n$ vertices $V$, partitioned into $k$ parts $V_1,\ldots,V_k$. 
    
    Let $\gamma_u(k)=k-\lceil k/u\rceil +1$.
    
    In $O(n^{\gamma_u(k)})$ time we can create a $\gamma_u(k)$-uniform hypergraph $G'$ on the same node set $V$ as $G$, so that $G'$ contains an $k$-hypercycle if and only if $G$ contains an $k$-hyperclique with one node from each $V_i$.
    
    If $G$ has weights on its hyperedges in the range $[-W,W]$, then one can also assign weights to the hyperedges of $G'$ so that a minimum weight $k$-hypercycle in $G'$ corresponds to a minimum weight $k$-hyperclique in $G$ and every edge in the hyperclique has weight between $[-\binom{\gamma_u(k)}{u}W,\binom{\gamma_u(k)}{u}W]$. Notably, $\binom{\gamma_u(k)}{u}\leq O(k^u)$.
\end{theorem}

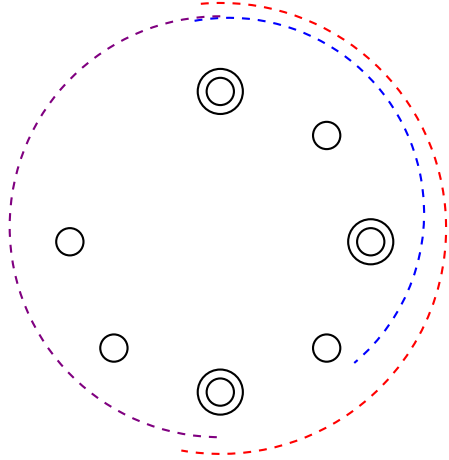
\begin{figure}[ht]
    \centering

\begin{tikzpicture}[scale=1]

\node[circle, draw, thick] (n1) at (90:2) {};
\node[circle, draw, thick] (n2) at (45:2) {};
\node[circle, draw, thick] (n3) at (0:2) {};
\node[circle, draw, thick] (n4) at (315:2) {};
\node[circle, draw, thick] (n5) at (270:2) {};
\node[circle, draw, thick] (n6) at (225:2) {};
\node[circle, draw, thick] (n7) at (180:2) {};

\draw[thick] (n1) circle (0.3);
\draw[thick] (n3) circle (0.3);
\draw[thick] (n5) circle (0.3);

\draw[dashed, thick, blue] (n1) ++(110:1) arc[start angle=100, end angle=-50, radius=2.6];

\draw[dashed, thick, violet] (n1) ++(90:1) arc[start angle=90, end angle=270, radius=2.8];

\draw[dashed, thick, red] (90:2.2) ++(105:1) arc[start angle=95, end angle=-100, radius=3];

\end{tikzpicture}
    \caption{The case of $k=7$ where we start with a $3$-uniform graph. The double circles indicate the `farthest apart' three nodes can be. Note the purple and blue arcs indicate the furthest a uniformity $4$ hyperedge can go while including the top node, and that neither hyperedge includes all three double circled nodes. Further note that the red edge, which is a hyperedge of uniformity $5$, covers all three double circled nodes. This is what $\gamma(\cdot)$ is capturing. }
    \label{fig:smallExampleGamma}
\end{figure}

To give the reader context and intuition we will explain the core ideas of this previous work. The main idea behind this reduction is to output $k$-partite hypergraphs such that the vertex partitions can be arranged in a circle and hyperedges only exist between adjacent partitions.
Then, the value $\gamma_u(k)$ corresponds to the smallest edge uniformity such that this type of edge will cover any possible subset of $u$ vertex partitions.
This allows this more restricted type of hyperedge to capture information about any possible hyperedge in the original $u$-uniform hypergraph. See Figure \ref{fig:smallExampleGamma} for a small illustrative example.

While this reduction shows a relationship between the hardness of hypercycle detection and hypercliques, it is important to observe that it requires a very specific uniformity $\gamma_u(k)$.
A natural question to ask is: what happens for other combinations of $u$ and $k$?

We show that the hardness varies based on the relationship between uniformity $u$ and cycle length $k$.
Note that for fixed $u$ and sufficiently large constant $k$, algorithms for finding or counting hypercycles yield runtimes independent of $k$.
This can be seen as a ``plateau'' in the hardness of hypercycle detection for any fixed $u$ seen in Figure \ref{fig:unweightedHypercycleBounds} and Section \ref{sec:ub_unweighted}.
For all $k$ where brute force is the best known algorithm, we give a matching lowerbound.


\subsection{Tight Hardness for Short Hypercycles}
\label{subsec:tighthardnessshorthypercycles}

We begin by showing that for finding relatively short hypercycles, brute-force algorithms are optimal unless the $(u,k)$-hyperclique hypothesis is false.
Note that the reduction from Theorem \ref{thm:hclique_to_hcycle} preserves the number of vertices in the graph $G$.
We use this second fact throughout the proof of our lower bound for this parameter regime.

Now, what exactly do we mean by ``short hypercycles''? 
To formally define the range in which we get tight results, we must introduce some notation, which allows us to reason about the uniformities we show hardness for.

\begin{definition}
    Recall the function $\gamma_u(k)=k-\lceil \frac{k}{u} \rceil+1$ from Theorem \ref{thm:hclique_to_hcycle}.
    Then, for constant $c$ define:
    
    \[\gamma_c^{-1}(u) = \max \{ k : \gamma_c(k) = u\}\]
\end{definition}
This function will make it easier to discuss the hypercycle lengths for which the hyperclique reduction yields hardness.
These correspond to the range $k \in [u, \gamma_3^{-1}(u)]$, which we will sometimes refer to as \textit{short hypercycles}.
In particular, we show that improving over brute-force search for short hypercycles would yield faster algorithms for finding hypercliques.
More formally we will prove that:

\tightShortHypercycle*

This conditional lower bound follows from the following ideas, which we prove as intermediate lemmas.
First, the function $\gamma_3$ is monotonically increasing, so applying the hyperclique reduction for $k\leq \gamma_3^{-1}(u)$ yields a uniformity $u' \leq u$.
Furthermore, we show a self-reducibility property of the hypercycle problem, which is that algorithms solving $k$-hypercycle on a given uniformity can be used to solve it on smaller uniformities.
Putting these ideas together, we are able to show that any hardness given by the hyperclique reduction on uniformities $u'<u$ can be extended to $u$.

\subsubsection{Understanding the Hyperclique to Hypercycle Reduction}

We begin by showing monotonicity of $\gamma$.

\begin{lemma} \label{lem:gamma monotone}
    For all $c\geq2$, the function $\gamma_c(k)$ is monotonically increasing.
\end{lemma}
\begin{proof}
    We will show that for any $k$, $\gamma_c(k+1) \geq \gamma_c(k)$.
    Note that $\gamma_c(k+1)=(k+1)-\lceil \frac{k+1}{c} \rceil +1$.
    Then, since $\lceil \frac{k+1}{c} \rceil \leq \lceil \frac{k}{c} \rceil+1$, we get $\gamma_c(k+1) \geq (k+1)-(\lceil \frac{k}{c} \rceil+1)+1=\gamma_c(k)$.
\end{proof}

\begin{corollary}
For every $k$ such that $u\leq k<\gamma_3^{-1}(u):\gamma_3(k) \leq u$.
\end{corollary}

We next want to show that if finding hypercycles is hard on $u$-uniform graphs, it is also hard in graphs with uniformity $u'>u$. 
To do this, we will make use of $k$-circle-layered hypergraphs:

\begin{definition}[$k$-circle-layered hypergraphs]
    We say that a hypergraph $G$ is $k$-circle-layered if it is $k$-partite and its vertex partitions $V_1, \cdots, V_k$ can be arranged in a circle such that hyperedges only exist between adjacent vertex partitions. That is between $u$ partitions $V_i, V_{i+1 \mod k},...,V_{i+u-1 \mod k}$. 
\end{definition}

While this seems like a highly restrictive definition, \cite[Lemma~2.2]{LVW18} shows that a time-$O(T(n,m))$ algorithm for finding $k$-hypercycles in this type of hypergraph can be used to find $k$-hypercycles in arbitrary hypergraphs using time $\tilde{O}(k^k(n+m+T(n,m)))$.
Thus, for practical purposes we will treat the problems as equivalent.
Moreover, when indexing into a vertex $v_i$ in a $k$-hypercycle or a partition $V_i$ in a $k$-circle-layered graph, the index $i$ should be interpreted as shorthand for $(i \mod{k})$.


Now, we need one more property of $k$-circle-layered hypergraphs before we prove our self-reducibility result.
The structure of these graphs are useful for reasoning about hypercycles when we can ensure that any hypercycle in the graph uses exactly one vertex from each partition.
When this is true, we can think of the hypercycle as ``going around'' the circular structure of the hypergraph.
Nonetheless, there are some scenarios in which more than one vertex from some partition may appear in the hypercycle.
We can think of these as ``backward cycles'', since they must turn around at some point to re-visit the repeated partition.

We will now show that such hypercycles cannot exist when $k \not\equiv 0 \mod{u}$.

\begin{lemma}\label{lem:no_short_backward_cycles}
Suppose there exists a $k$-hypercycle $v_0,...,v_{k-1}$ in a $k$-circle layered hypergraph with uniformity $u$ that does not use exactly one node from each partition. Then, we have that $k\equiv 0 \mod{u}$.
\end{lemma}
\begin{proof}
We begin with the observation that such a hypercycle must have at least one partition where it does not have a node. Locate one partition such that it has a node from the cycle, but where one of its adjacent partitions does not. WLOG, label this partition $V_0$, and the adjacent partition that does not have a vertex from the cycle $V_{k-1}$, and the vertices between them $V_1$ to $V_{k-2}$ in order. This means there is no edge in the cycle that contains both a vertex from $V_0$ and $V_{k-1}$.

This choice allows us to make another observation. For every single edge in the cycle, it must be the case that there is exactly one vertex within it that comes from a partition where the label is a multiple of $u$. This is because all such partitions in consideration for the cycle are at least a distance of $u$ away from each other, and no edge spans that distance around the circle. However, every edge must also intersect the partition somewhere, as there are only $u-1$ partitions between two partitions labeled with multiples of $u$.

This allows us to make the following claim: For all $a$ such that $au < k-1$, $v_{au}\in V_{ru}$ for some $r\in \mathbb Z$.

We prove this claim using induction. The base case is trivial. Given that the statement is true for $a$, we aim to show that it is true for $a+1$. We have that $v_{au}, v_{au + 1}..., v_{a(u+1)-1}, v_{a(u+1)}$ is a sequence in the hypergraph, and that the first $u$ vertices in the sequence form an edge, as do the last ones. This means that $v_{au + 1}..., v_{a(u+1)-1}$ cannot contain a vertex from a partition with a multiple of $u$, and thus $v_{a(u+1)}$ must come from such a partition.

This proves the claim. Now, towards a contradiction, we suppose that $k$ is not a multiple of $u$. We observe the hyperedge $v_{k-u+1}, ..., v_{k-1}, v_0$. We note that there must be some vertex in the first $u-1$ vertices that has a label that is a multiple of $u$. Call this $v_i$. We also note that this hyperedge must span the partitions $V_0,...V_{u-1}$, as no vertex is in $V_{k-1}$. This, combined with the claim, implies that $v_i$ and $v_0$ are both in $V_0$, a contradiction.
\end{proof}

\subsubsection{Increasing Uniformity Preserves Hardness}

We can now show the following self-reducibility lemma about finding and counting $k$-hypercycles:

\begin{lemma} \label{lem:uniform self reduction}
    Let $k$ be a hypercycle length such that $k \not \equiv 0 \mod{u}$.
    Suppose there exists $\varepsilon>0$ such that there is an algorithm for finding (counting) $(u,k)$-hypercycles in $k$-circle-layered graphs that runs in time $O(n^{k-\varepsilon})$.
    Then, for any uniformity $u'<u$, there exists an algorithm for finding (counting) $(u',k)$-hypercycle in $k$-circle-layered graphs running in time $O(n^{k-\varepsilon}+n^u)$.
\end{lemma}
\begin{proof}
    We will show how to map an $n$-node $u'$-uniform hypergraph $G=(V,E)$ to an $n$-node $u$-uniform hypergraph $G'$ such that $G'$ has a $k$-hypercycle if and only $G$ has one.

    We construct $G'$ as follows.
    We let its vertex set $V'=V$, and we will add edges so that $G'$ is also $k$-circle-layered.
    Then, we construct hyperedges in $G'$ as follows.

    For each hyperedge $(v_i,v_{i+1}, \cdots , v_{i+u'-1})$, we create $n^{u-u'}$ hyperedges by creating all possible ``extensions'' of the hyperedge to the next $u-u'$ partitions.
    More concretely, we know that the vertices of this hyperedge satisfy $v_i \in V_i, v_{i+1} \in V_{i+1}, \cdots, v_{i+u'-1} \in V_{i+u'-1}$.
    Then, for all possible combinations of vertices $y_{i+u'} \in V_{i+u'}, \cdots, y_{i+u'-1} \in V_{i+u-1}$ we will add the hyperedge $(v_i,\cdots,v_{i+u'-1},y_{i+u'},\cdots, y_{i+u-1})$ to $E$.

    This mapping results in a hypergraph $G'$ such that $|V'|=|V|=n$ and $|E'|=|E|^{u-u'}$.
    This blowup in the number of edges will not be a problem since the algorithm we are assuming for $(u,k)$-hypercycle is agnostic to sparsity.

    Now we must show that $G'$ preserves $k$-hypercycles from $G$ and does not create any new ones.
    First, assume there exists a hypercycle $v_1, v_2, \cdots, v_k$ in $G$.
    Then, for all consecutive sets of $u'$ vertices, we have that the hyperedge $(v_{i},v_{i+1},\cdots, v_{i+u'-1}) \in E$.
    Since we added all possible extensions of this hyperedge to $E'$, this necessarily implies that for the specific extension corresponding to the next $u-u'$ vertices of the hypercycle, the hyperedge $(v_i,v_{i+1},\cdots,v_{i+u'-1},v_{i+u'}, \cdots, v_{i+u-1}) \in E$.
    By applying this reasoning to all of the edges in the hypercycle, we get that $v_i, \cdots, v_k$ also form a $k$-hypercycle in $G'$.

    Now we show why $G'$ does not contain any hypercycles that cannot be traced back to a hypercycle in $G$.
    The key idea for this will be the fact that for every hyperedge in $G'$, the first $u'$ nodes in the hyperedge form a hyperedge in $G$.
    
    Assume there exists a hypercycle $v_1, \cdots, v_k$ in $G'$.
    We want to show that these vertices also form a hypercycle in $G$, which means we want to show that for each consecutive set of $u'$ vertices we have $(v_i,\cdots, v_{i+u'-1}) \in E$.
    Now, we know that for any set of consecutive $u'$ vertices in the hypercycle, there is a hyperedge in $G'$ of the form $(v_i,\cdots, v_{i+u'-1},v_{i+u'},\cdots, v_{i+u-1})$.
    This is due to the fact that, by Lemma \ref{lem:no_short_backward_cycles}, any $k$-hypercycle in this parameter regime must use exactly one vertex from each partition.
    Then, if this edge was added to $E'$, this means the first $u'$ vertices must form a hyperedge in $G$.
    Consequently, $v_1, \cdots, v_k$ form a $k$-hypercycle in $G$.

    Now that our mapping is complete, we can specify how an algorithm $A$ solving $(u,k)$-hypercycle in $n^{k-\varepsilon}$ time can be used to solve $(u',k)$-hypercycle in the same runtime.
    Given an $n$-node, $u'$-uniform hypergraph $G$, we can compute a $u$-uniform hypergraph $G'$ following our reduction.
    This requires iterating through all hyperedges in $G$ then creating all $n^{u-u'}$ possible extensions and there can be up to $O(n^{u'})$ hyperedges, so the reduction takes $O(n^u)$ time.
    Then, since $G'$ also has $n$ nodes, we can run $A$ on $G'$ to determine whether there is a $k$-hypercycle in $O(n^{k-\varepsilon})$ time.

    The total runtime of this algorithm is $O(n^{k-\varepsilon}+n^u)$, which is what we wanted to show.
\end{proof}

\subsubsection{Getting the Tight Lower Bound}

We finally prove that short hypercycles require brute force:

\tightShortHypercycle*
\begin{proof}
    We want to show that for all $k \in [u, \gamma_3^{-1}(u)]$, finding a $(u,k)$-hypercycle requires $n^{k-o(1)}$ time.
    We will do this by applying Theorem \ref{thm:hclique_to_hcycle} to go from hardness of $(3,k)$-hyperclique to hardness of $(\gamma_3(k),k)$-hypercycle, then use monotonicity of $\gamma_3$ and the self-reducibility lemma to get the desired hardness for $(u,k)$-hypercycle.

    First, observe that the reduction from $(3,k)$-hyperclique to $(\gamma_3(k),k)$-hypercycle preserves the number of vertices in the hypergraph, and this reduction can be computed in time $O(n^{\gamma_3(k)})$.
    Moreover, this reduction outputs a $k$-circle-layered hypergraph.
    Consequently, an algorithm solving $(\gamma_3(k),k)$-hypercycle in $n^{k-\varepsilon}$ time would solve $3,k$-hyperclique in the same runtime.
    This would violate the $(3,k)$-hyperclique conjecture, so we get a conditional lower bound for $(\gamma_3(k),k)$-hypercycle.
    Now, since $k \leq \gamma_3^{-1}(u)$, Lemma \ref{lem:gamma monotone} tells us that $\gamma_3(k) \leq u$.
    In the case these are equal, we already have the desired lower bound.
    
    If the inequality is strict, we know from Lemma \ref{lem:uniform self reduction} that an $O(n^{k-\varepsilon})$ algorithm for $(u,k)$-hypercycle would imply an algorithm with the same runtime for $(\gamma_3(k),k)$-hypercycle since $k\leq \gamma_3^{-1}(u)<2u$ so it is not a multiple of $u$.
    This violates our conditional lower bound, so we can conclude that under the $(3,k)$-hyperclique assumption, finding a $(u,k)$-hypercycle requires $n^{k-o(1)}$ time.
\end{proof}

From out worst-case to average-case reduction in Lemma \ref{lem:WCtoAC}, we can conclude that the average-case version of counting $(u,k)$-hypercycles also requires brute-force:

\begin{corollary}
    If the $(3,k)$-hyperclique hypothesis holds then counting $k$-hypercycles in $u$-uniform Erd{\H{o}}s-R{\'{e}}nyi hypergraphs when $k \in [u, \gamma_3^{-1}(u)]$ requires $\tilde{O}(n^{k-o(1)})$ time for constant $u$.
\end{corollary}
\begin{proof}
    If $u$ is constant then so is $k$. Thus the number of edges and nodes in our subhypergraph is constant. 

    Recall we are limiting ourselves to $k \in [u, \gamma_3^{-1}(u)]$.
    By Theorem \ref{thm:tight short hypercycle} we have that deciding if a $k$-hypercycle exists in a $k$-partite graph is $\tilde{O}(n^{k-o(1)})$ hard if the $(3,k)$-hyperclique hypothesis is true. 

    We can then apply Lemma \ref{lem:WCtoAC} to conclude that counting $k$-hypercycles in $u$-uniform Erd{\H{o}}s-R{\'{e}}nyi hypergraphs is $\tilde{O}(n^{k-o(1)})$ hard as well.
    
\end{proof}

\section{Beating Brute Force for Longer Hypercycles}
\label{sec:ub_unweighted}

In this section we give a new faster algorithm for hypercycle. 
While the range in which we get tightness in the previous section may have seemed arbitrary, the brute-force lower bound is actually \textit{false} for any longer hypercycle.
In particular, we show two different algorithms which leverage matrix multiplication to beat brute force when $k$ is sufficiently bigger than $u$.
One of these algorithms allows improvement over $n^{k-o(1)}$ starting at exactly $k \geq \gamma_3^{-1}(u)+1$, which is the point at which we can no longer prove tightness. 
The second algorithm exploits the sparsity of its input, and yields even better runtimes than the first when $k \geq 2u-1$.

These algorithms suggest the following relationship between hypercycle length and potential lower bounds.
When $k \in [\gamma_3^{-1}(u),2(u-1)]$, hardness is still dominated by the hypercycle length, but any reasonable lower bounds must account for fast matrix multiplication.
Then, when $k \in [2u-1, \infty]$, since the algorithm exploits sparsity and the total number of possible hyperedges is $O(n^u)$, hardness becomes independent of hypercycle length and is dominated by uniformity. 

\subsection{Faster Algorithm via Triangle-Detection} \label{subsec:triangle_detection}

We now show how matrix multiplication can be used to speed up hypercycle algorithms.
The main idea behind this algorithm is that when $k>\gamma_3^{-1}(u)$, it is possible to find three vertex partitions that are sufficiently ``spread out'' that no hyperedge covers all three partitions. This lets us treat the relationship between these three partitions as edges, not hyperedges. Specifically, we can represent the problem with many $3$-partite graphs which contain a triangle if and only if the original hypergraph contains a $k$-hypercycle. 
Then, since triangle-detection can be solved in $n^\omega$ time using matrix multiplication, we can get some savings.

To present our algorithm, we will need the following lemma.

\begin{lemma} \label{lem: triangle partitions}
    Let $\delta = \lceil k/3 \rceil$.

    Let $G$ be an $n$-vertex, $u$-uniform $k$-circle layered hypergraph with vertex partitions $V_1, \cdots, V_k$.
    Then, for all $k>\gamma_3^{-1}(u)$ and $i \in [1,k]$, every hyperedge covers at most two out of the three vertex sets $V_i,V_{i+\delta},V_{i+2\delta}$.
\end{lemma}
\begin{proof}
    Without loss of generality let $i =1$. 
    
    Recall that $\gamma_3^{-1}(u)$ is the maximum hypercycle length such that the hyperclique reduction outputs uniformity $u$ when starting with a $3$-uniform graph.
    Thus, if $k>\gamma_3^{-1}(u)$, applying the reduction to $(3,k)$-hyperclique yields a uniformity $u'>u$.
    Note that the three vertex sets $V_1,V_{1+\delta},V_{1+2\delta}$ are chosen to be maximally spaced out, meaning that if we replace one of them with any other partition, the smallest arc covering the new triple will be shorter than the one covering these three.

    We have that $u < k - \lceil k/3 \rceil +1$, because we are choosing a $k$ above $\gamma_3^{-1}(u)$. Note that any hyperedge which covers all three partitions must cover all the partitions except for those between some pair of partitions. So, the question is what is this gap? First, the gap between $V_1$ and $V_{1+\delta}$ is $\delta-1$ partitions. Second, the gap between $V_{1+\delta}$ and $V_{1+2\delta}$ is $\delta-1$ partitions. Finally, the gap between $V_{1+2\delta}$ and $V_1$ is $k-2\delta-1$ partitions. To cover all three partitions a hyperedge must cover all $k$ partitions except those in the gap. So, no hyperedge covers all three partitions as long as 
    $$u < \min(k-\delta+1, 2\delta+1).$$
    We chose $\delta$ to be $\lceil k/3 \rceil$. So
    $$u < k - \lceil k/3 \rceil +1 = k-\delta +1 = k -  \lceil k/3 \rceil +1. $$
    Now we need to check the second case how does 
    $u < k - \lceil k/3 \rceil +1$ compare to $2 \lceil k/3 \rceil + 1$? We can quickly work through the three cases. If $k$ is a multiple of 3 
    $$u < 2k/3+1 = 2k/3 +1.$$
    If $k$ is congruent to $1$ mod 3 then 
    $$u < k - (k+2)/3 +1 = 2k/3 + 1/3 < 2k/3 + 7/3 = 2 \cdot (k+2)/3 +1.$$
    If $k$ is congruent to $2$ mod 3 then 
    $$u < k - (k+1)/3 +1 = 2k/3 + 2/3 < 2k/3 + 5/3 = 2 \cdot (k+1)/3  +1.$$

    So, regardless of the congruence class of $k$ if $\delta = \lceil k/3 \rceil$ then these partitions are not covered by a single hyperedge.

    
\end{proof}

Our algorithm uses this idea to combine brute-force search with known triangle-counting techniques, which yields a saving of $(3-\omega)$ in the exponent of the brute-force runtime.

\begin{restatable}{theorem}{matmultCycle} \label{thm: matmul_cycle}
    Let $G$ be an $n$-vertex $u$-uniform $k$-circle-layered hypergraph and $k > \gamma^{-1}_3(u)$.
    Then, there exists an $O(n^{k-3+\omega})$-time algorithm for finding (counting) $k$-hypercycles in $G$.
\end{restatable}

\begin{proof}

    Let $\delta = \lceil k/3 \rceil$.
    
    Let $V_1, \cdots, V_k$ be the vertex partitions of $G$.
    The algorithm proceeds as follows.

    First, fix three vertex sets $V_1,V_{1+\delta},V_{1+2\delta}$, which we call \textit{triangle partitions}.
    These satisfy that no hyperedge covers all three partitions, because of Lemma \ref{lem: triangle partitions}.
    Now, we will use brute force to search over the remaining $k-3$ non-triangle partitions. For convince we call the brute force chosen vertex from partition $V_i$ $b_i$.
    In particular, for all $n^{k-3}$ possible combinations of vertices, one from each partition, we do the following:
    
    Create a tri-partite graph $G'$ with three vertex sets $T_1=V_1,T_2=V_{1+\delta},T_3=V_{1+2\delta}$.
    Then, there will exist an edge between two vertices if and only if there is a hyperpath between them that goes through the current choice of vertices. In other words, we only include an edge between a node in e.g. $v_ \in T_1$ and $v_2 \in T_2$ if every edge which involves only brute forced nodes $b_i$, $v_1$, and $v_2$ exists in the original hypergraph\footnote{For a concrete example consider $u=3$, $k=5$, so $\delta = 2$ where the triangle partitions are $V_1, V_3, V_5$. There will be an edge between $v_1$ and $v_3$ iff hyperedges $(v_1,b_2,v_3)$, $(b_2,v_3,b_4)$ exist in the original graph.}.
    This graph can be constructed in $O(n^2)$ time since we need to look at all pairwise combinations of the three vertex sets.
    Finally, use matrix multiplication to run triangle detection (or counting) on $G'$ which takes $O(n^{\omega})$. 

    If the algorithm finds a triangle for any of the $n^{k-3}$ graphs, then there is a hypercycle in $G$ and we output YES.

    We complete this proof by analyzing the runtime of the algorithm.
    For each combination of vertices, we take $O(n^2+n^{\omega})$ time.
    Since this must be done for $O(n^{k-3})$ combinations, our total runtime is $O(n^{k-3+\omega})$.
\end{proof}
This algorithm shows that for hypercycles which are sufficiently longer than the hyperedge size, it is actually possible to beat brute force.
It is worth noting that the algorithm's complexity still increases with $k$.
This leaves open the possibility that we can find reasonably tight lower bounds dependent on $k$ for $k \geq \gamma^{-1}_3(u)$.

Nonetheless, we next show that such a lower bound is only possible for hypercycles up to length $2(u-1)$, since there is an algorithm which only depends on $u$ and is faster for $k \geq 2u-1$.

\subsection{A Faster Algorithm for Longer Hypercycles} \label{subsec:longer_hypercycles}

In this section, we explore a different approach to the hypercycle problem.
Once again, we exploit the structure of a $k$-circle-layered graph.
The main idea behind previous algorithms for e.g. weighted cycle is to start at some partition $V_1$, then for each vertex $v_1$ in $V_1$ iterate through $V_2,\cdots,V_k$ while keeping track of whether $v_1$ has a path to each of these \cite{LVW18}.
When going from $V_i$ to $V_{i+1}$, a node in $V_{i+1}$ will be reachable if it shares an edge with a node in $V_i$ that is reachable from $v$.
This reduces the cycle problem to this reachability subproblem between adjacent partitions.

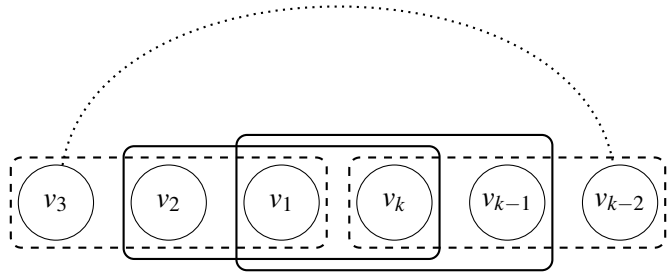
\begin{figure}[ht]
    \centering
    
\begin{tikzpicture}[scale=1.5, every node/.style={circle, draw, minimum size=1cm, inner sep=0pt}]

\node (v3) at (0,0) {$v_3$};
\node (v2) at (1,0) {$v_2$};
\node (v1) at (2,0) {$v_1$};
\node (vk) at (3,0) {$v_k$};
\node (vk1) at (4,0) {$v_{k-1}$};
\node (vk2) at (5,0) {$v_{k-2}$};

\draw[thick, rounded corners] (0.6,0.5) rectangle (3 + 0.4,-0.5);
\draw[thick, rounded corners] (1.6,0.6) rectangle (4 + 0.4,-0.6);
\draw[thick, dashed, rounded corners] (2.6,0.4) rectangle (5 + 0.4,-0.4);
\draw[thick, dashed, rounded corners] (-0.4,0.4) rectangle (2 + 0.4,-0.4);

\draw[dotted, thick, bend left=80] (v3) to (vk2);

\end{tikzpicture}
    \caption{Example with uniformity $u=3$ to demonstrate why you need $u-1$ nodes. These are the four hyperedges which involve only the vertices $v_1,v_2, v_2, v_k, v_{k-1},$ and $v_{k-2}$. The two thick hyperedges include at least one vertex from $v_1, v_2, v_3$ and at least one vertex from $v_k,v_{k-1},v_{k-2}$, the dashed hyperedges include only vertices from one side. Note that only $v_2, v_1, v_k,$ and $v_{k-1}$ are involved in the crossover edges. Note that $2 = u-1$ in this context.}
    \label{fig:reachabilityHypercycle}
\end{figure}

We define a more general reachability subproblem, modified in ways that are necessary to use it for hypercycle detection. 
The key difference is that instead of considering all possible starting vertices for the hypercycle, we have to consider all possible sets of $u-1$ starting vertices.
This is due to the fact that at the end of the hypercycle, we will need the last $u-1$ edges to contain all of these vertices. See Figure \ref{fig:reachabilityHypercycle} for a visual representation of why $u-1$ partitions are needed to finish the cycle. 

We define a reachability problem in circle-layered graphs to capture this idea.

\subsubsection{Reachability Problems in Uniform Hypergraphs}

\newcommand{\uCLR}{$u$-CLR}

Intuitively the following problem asks for all hyperpaths that cross $2u-1$ partitions in a $(2u-1)$-circle-layered hypergraph. The output will answer for all tuples of nodes of the first $u-1$ nodes and last $u-1$ nodes if a path exists, or what the minimum weight path is. 

\begin{definition}[(Weighted) $u$-uniform circle-layered reachability problem ((Weighted) \uCLR)] \label{def: u-CLR}
Let $G=(V,E)$ be an $n$-vertex $u$-uniform $(2u-1)$-circle-layered (weighted) hypergraph with vertex partitions $V_1, \cdots, V_{2u-1}$. 
Let $L=V_1 \times \cdots \times V_{u-1}$ and $R=V_{u+1} \times \cdots \times V_{2u-1}$ be sets of vertex partitions.
Then, the $u$-uniform circle-layered reachability problem asks for an output of size $n^{2u-2}$. 
For all $2u-2$ tuples of nodes $(x_1, \cdots, x_{u-1}, y_1, \cdots, y_{u-1})$  where $(x_1, \cdots, x_{u-1}) \in L$ and  $(y_1, \cdots, y_{u-1}) \in R$ you must return the following:
\begin{itemize}
    \item If unweighted return if there exists a $v \in V_{u}$ 
such that edges $(x_1, \cdots \, x_{u-1},v)$, $(x_2, \cdots \, x_{u-1},v,y_1)$,  $\ldots$, $(x_{u-1},v, y_1, \cdots, y_{u-2})$, $(v, y_1, \cdots, y_{u-1})$ exist  in $E$. (That is if a hyperpath $x_1, \cdots, x_{u-1}, v, y_1, \cdots, y_{u-1}$ exists.)  
    \item If weighted let $w()$ be a function which returns the weight of a hyperedge. Then return 
    $$ min_{v\in V_{u}}  \left( w(x_1, \cdots \, x_{u-1},v)+ w(v, y_1, \cdots, y_{u-1}) + \sum_{i\in [2,u-2]} w(x_i,x_{i+1},\ldots,x_{u-1}, v, y_1, \ldots, y_{i-1}) \right) $$
\end{itemize}
  
The output of this problem is an $n^{u-1} \times n^{u-1}$ matrix, which we denote CLR($L,R$) (we will also call this output matrix CLR$_{2u-1}$($L,R$)). This matrix is binary if unweighted and is over field of the weights if weighted. 
\end{definition}

For convenience we will also define a second version of this which takes in a CLR$_{k-1}$($L,R$) matrix and outputs another reach-ability matrix, CLR$_{k}$($L,R$),  for a graph with one more circle layer. The key intuition for why we define the problem this way is that if you have reachability information from the first $u-1$ nodes to the last $u-1$ nodes you can extend this to another partition because no hyperedge in the path will need information from the `middle' nodes as the hyperedges are of uniformity $u$.
\newcommand{\uECLR}[2]{$(#1,#2)$-ECLR}

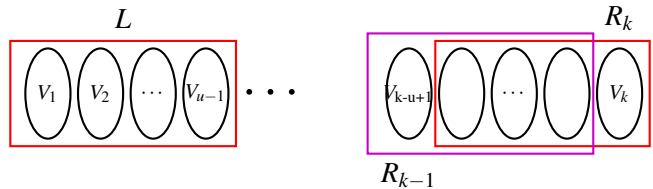
\begin{figure}[ht]
    \centering
    \begin{tikzpicture}

    \draw[thick, red] (-3.5,1.2) rectangle (-0.5,-0.2);
    \foreach \i/\label in {0/V_1, 1/V_2, 2/\cdots, 3/V_{u-1}} {
        \draw[thick] (-3+\i*0.7,0.5) ellipse (0.3 and 0.6);
        \node at (-3+\i*0.7, 0.5) {\scriptsize $\label$};
    }
    \node at (-2, 1.5) {$L$};

    \node at (0, 0.5) {\huge $\cdots$};

    \draw[thick, red] (2.15,1.2) rectangle (5,-0.2);
    \foreach \i/\label in {-1/V_{\text{\tiny{k-u+1}}}, 0/, 1/\cdots, 2/, 3/V_{k}} {
        \draw[thick] (2.5+\i*0.7,0.5) ellipse (0.3 and 0.6);
        \node at (2.5+\i*0.7, 0.5) {\scriptsize $\label$};
    }
    \node at (4.6, 1.5) {$R_k$};

    \draw[thick, purple] (1.25,1.3) rectangle (4.25,-0.3);
    \node at (1.8, -0.6) {$R_{k-1}$};

    \end{tikzpicture}
    \caption{The \uECLR{u}{k} problem.}
    \label{fig:ECLR}
\end{figure}

\begin{definition}[(Weighted) $(k,u)$-extension uniform circle-layered reachability problem ((Weighted) \uECLR{u}{k})] \label{def: uk-ECLR}
Let $G=(V,E)$ be an $n$-vertex $u$-uniform $k$-circle-layered (weighted) hypergraph with vertex partitions $V_1, \cdots, V_{k}$. 
Let $L=V_1 \times \cdots \times V_{u-1}$ and $R_k=V_{k-(u-1)+1} \times \cdots \times V_{k}$ be sets of vertex partitions.
Additionally let CLR$_{k-1}(L, R_{k-1})$ be a reachability matrix which gives (minimum weight) hyperpath reachability from $L$ to $V_{k-(u-1)-1} \times \cdots \times V_{k-1}$.

Then, the \uECLR{u}{k} problem asks for an output of size $n^{2u-2}$. 
For all $2u-2$ tuples of nodes \\
$(x_1, \cdots, x_{u-1}, y_1, \cdots, y_{u-1})$  where $(x_1, \cdots, x_{u-1}) \in L$ and  $(y_1, \cdots, y_{u-1}) \in R$ you must return the following:
\begin{itemize}
    \item If unweighted return if there exists a hyperpath which starts at nodes $x_1, \cdots, x_{u-1}$ and ends at nodes $ y_1, \cdots, y_{u-1}$. 
    \item If weighted then return the minimum weight hyperpath which starts at $x_1, \cdots, x_{u-1}$ and ends at nodes $ y_1, \cdots, y_{u-1}$. 
\end{itemize}
  
The output of this problem is an $n^{u-1} \times n^{u-1}$ matrix, which we denote CLR$_k$($L,R$). This matrix is binary if unweighted and is over field of the weights if weighted. See Figure \ref{fig:ECLR} for a visual of this problem. 
\end{definition}

\subsubsection{Hypercycle Algorithm from Reachability Algorithm}

We now show that a pair of algorithms solving these problems can be used efficiently to solve hypercycle.

\begin{lemma} \label{lem: dp-like algorithm}
    
Suppose there exists a $T_1(n)$-time algorithm for \uCLR~and a $T_2(n)$-time algorithm for \uECLR{u}{k}. 
Then, there exists an $O(T_1(n)+T_2(n)+n^{2(u-1)}))$-time algorithm for finding (counting) $k$-hypercycles in $n$-vertex $u$-uniform $k$-circle-layered hypergraphs.
\end{lemma}

\begin{proof}
    Let $A_1$ be the algorithm for \uCLR, $A_2$ be the algorithm for \uECLR{u}{k} and let $V_1, \cdots, V_k$ be the vertex partitions of the $k$-circle-layered hypergraph.
    Let $L$ and $R_i$ be as defined in Definition \ref{def: uk-ECLR}.

    Note that given the matrix \uCLR$_k$($L,R_k$), we can solve hypercycle by ``completing the hyperpath'' as follows.
    For each possible sequence $(x_1, \cdots, x_{u-1})$ in $L$, we iterate through each possible sequence $(y_1, \cdots, y_{u-1})$ in $R_k$ and check its reachability.
    If the value in the matrix is a $1$, we can then check whether all of the hyperedges $(y_1, \cdots, y_{u-1},x_1), \cdots , (y_{u-1},x_1, \cdots, x_{u-1})$ exist.
    If they do, we have found a hypercycle. 
    
    Note that this procedure requires $O(1)$ work for each possible pair of sequences, and there are $n^{2(u-1)}$ pairs, so this requires $O(n^{2(u-1)})$ time.

    Now we must show how to efficiently obtain \uCLR$_k(L,R_k$) using the algorithms $A_1$ and $A_2$.
    First, we use $A_1(L,R_{2u-1})$ to get \uCLR$_{2u-1}(L,R_{2u-1}$) in $T_1(n)$ time.
    Next, we can repeat the following for $i \in [2u,k]$: Run $A_2$ on inputs $L,R_i$ and \uCLR$_{i-1}(L,R_{i-1}$) to get \uCLR$_{i}(L,R_{i}$).
    Once we get to $i=k$, we can carry out the hyperpath completion mentioned previously and check for hypercycles. 

    This process requires running $A_1$ once and $A_2$ a constant $k-2u+1$ number of times, which takes $T_1(n)+(k-2u+1)T_2(n)=O(T_1(n)+T_2(n))$ time.
    Thus, the total runtime of the algorithm is $O(T_1(n)+T_2(n)+n^{2(u-1)}))$.
\end{proof}

Note that if $T_1(n)$ and $T_2(n)$ have runtimes independent of $k$, then so will this algorithm.
Inutitively this should be true since one problem is completely characterized by $(2u-1)$-circle-layered graphs, and the other is simply asking about extending the reachability by one layer.
In the next subsection, we actually prove this by constructing algorithms for \uCLR~and \uECLR{u}{k}~ which will let us prove the following theorem:

\begin{restatable}{theorem}{matMulReach}\label{thm:matmul_reachability}
There exists a time-$O(n^{2u-1-(3-\omega)})$ algorithm for finding (counting) $k$-hypercycles in $n$-node $u$-uniform $k$-circle layered hypergraphs when $k \geq 2u-1$.
\end{restatable}

\subsubsection{Algorithms for Unweighted Reachability}

To achieve the desired runtime, we need algorithms for \uCLR~and \uECLR{u}{k} running in time $O(n^{2u-1- (3-\omega}))$. 
We will show that this can be achieved using matrix multiplication in a similar way as done in Theorem \ref{thm: matmul_cycle}.

We first show how to do this for $k=2u-1$.
\begin{lemma} \label{lem: uECLR algorithm}
    There exists a time-$O(n^{2u-1-(3-\omega)})$ algorithm for solving \uCLR~ in $n$-node $u$-uniform hypergraphs.
\end{lemma}
\begin{proof}
    The algorithm heavily relies on the fact that any given hyperedge will cover at most two out of the three vertex partitions $V_1,V_u,V_{2u-1}$.
    We fix these three, then for each of the $n^{2u-1-3}$ possible combination of vertices from the other $2u-1-3$ vertex partitions, we will create a graph $G$ which captures the reachability between these three partitions.

    The graph will have three vertex partitions, corresponding to $V_1, V_u, V_{2u-1}$, with edges between $V_1$ and $V_u$ as well as between $V_u$ and $V_{2u-1}$.
    We add an edge between any pair of vertices in $V_1$ and $V_u$ which has a hyperpath going through the current choice of vertices in $V_2, \cdots, V_{u-1}$.
    Similarly, we add an edge between a pair of vertices in $V_u$ and $V_{2u-1}$ if there is a hyperpath between them going through the current choice of vertices in $V_{u+1}, \cdots, V_{2u-2}$.

    Once the graph has been constructed, assume we have two adjacency matrices, one for each pair of vertex sets with edges between them.
    Then, we can use matrix multiplication to obtain the reachability between $V_1$ and $V_{2u-1}$ in $O(n^{\omega})$ time.
    Since we do this for $n^{2u-1-3}$ graphs, the total runtime of the algorithm is $O(n^{2u-1-(3-\omega)})$ time.
\end{proof}

We now show how to extend the technique to solve reachability when $k>2u-1$.

\begin{lemma} \label{lem: uCLR algorithm}
    There exists a time-$O(n^{2u-1-(3-\omega)})$ algorithm for solving \uECLR{u}{k} in $n$-node $u$-uniform hypergraphs.
\end{lemma}
\begin{proof}
    Recall that we want to produce the matrix \uCLR$_k(L,R_k)$ given $L,R_k$ and the reachability matrix \uCLR$_{k-1}(L,R_{k-1})$, where $L,R_{k-1},R_k$ all consist of $n^{u-1}$ tuples of $(u-1)$ vertices.
    The algorithm essentially creates a reachability matrix from $R_{k-1}$ to $R_k$, then combines this with \uCLR$_{k-1}(L,R_{k-1})$ to extend the hyperpath.

    To create the reachability matrix $A$ from $R_{k-1}$ to $R_{k}$, we must look at all possible pairs of $(u-1)$-tuples in $R_{k-1} \times R_k$.
    If the two tuples have the form $(x,v_2, \cdots, v_{u-1})$ and $(v_2, \cdots, v_{u-1},y)$ then we check whether the hyperedge $(x, v_2, \cdots, v_{u-1},y) \in E$.
    If it is, then we set the corresponding entry in $A$ to $1$.
    Note that doing this for all possible pairs takes $n^{2(u-1)}$ time.

    To obtain \uCLR$_k(L,R_k)$ from $A$ and \uCLR$_{k-1}(L,R_{k-1})$ we use the same approach as in the previous lemma.
    We will fix three vertex sets $V_1,V_{k-(u-1)},V_k$ and then create $2^{2u-1-3}$ different graphs by taking all possible combinations of vertices in $V_2, \cdots, V_{u-1},V_{k-u}, \cdots, V_{k-1}$.
    For each combination, the graph will have three vertex partitions corresponding to $V_1,V_{k-(u-1)},V_k$ and edges will be added between $V_1$ and $V_{k-(u-1)}$ as well as between $V_{k-(u-1)}$ and $V_k$.
    Crucially, this can be done because $k$ is big enough so that any given hyperedge will cover at most two of these 
    To determine reachability between $V_1$ and $V_{k-(u-1)}$, we can check the corresponding entry in \uCLR$_{k-1}(L,R_{k-1})$.
    Even though this matrix contains a stricter condition, this condition is necessary for the reachability to $V_k$, so we enforce it with this connection.
    To determine reachability between $V_{k-(u-1)}$ and $V_k$, we can use the corresponding entries in $A$.

    Once we have this graph, we can use matrix multiplication to determine 2-hop reachability and fill out the appropriate entry in \uCLR$_k(L,R_k)$ in $O(n^{\omega})$ time. 
    Since we do this for all $n^{2u-1-3}$ combinations of vertices, the total runtime of the algorithm is $O(n^{2u-1-(3-\omega)})$.
\end{proof}

We can now prove that $k$-hypercycle can be solved in $O(n^{2u-1-(3-\omega)})$ for sufficiently large $k$.

\matMulReach*

\begin{proof}
    Recall from Lemma \ref{lem: dp-like algorithm} that given a $T_1(n)$-time algorithm for \uCLR~and a $T_2(n)$-time algorithm for \uECLR{u}{k} we can solve $k$-hypercycle in time $O(n^{2(u-1)}+T_1(n)+T_2(n))$. 
    Applying Lemma \ref{lem: uCLR algorithm}, we know we can set $T_1(n)=O(n^{2u-1-(3-\omega)})$.
    From Lemma \ref{lem: uECLR algorithm} we get $T_2(n)=O(n^{2u-1-(3-\omega)})$.

    Combining these, we conclude that $k$-hypercycle in $u$-uniform hypergraphs can be solved in $O(n^{2u-1-(3-\omega)})$ time when $k \geq 2u-1$.
\end{proof}

We can now prove our desired theorem.

\fastCycleDetect*

\begin{proof}
    Using the color-coding approach from \cite[Lemma~2.2]{LVW18} we can randomly color the vertices in $G$ to obtain a $k$-circle-layered hypergraph.
    Then, we can run our time-$O(n^{2u-1-(3-\omega)})$ algorithm on this new hypergraph.
    Finally, this process can be repeated $k^k \log{n}$ times to boost the success probability, giving us our runtime of $\tilde{O}(k^k(n+m+n^{2u-1-(3-\omega)}))$.
\end{proof}

\section{Tight Results for Min-Weight Hypercycle}
\label{sec:tight_weighted}
In this section we present the tight matching upper and lower bounds for minimum hypercycle represented in Figure \ref{fig:weightedHypercycleBounds}. We start by providing the algorithm for min-weighted hypercycle and then prove the lower bounds.

Notably, we demonstrate that the best algorithm for weighted $k$-hypercycle for $k \in [u+1, 2u-1]$ is the naive algorithm. Our algorithm relies on tracking the first $u-1$ nodes of a hyperpath and the last $u-1$ nodes of a hyperpath.  This lower bound shows that doing so is, in some sense, required. The tight upper and lower bound gives credence to the intuition that when thinking about hypercycles you should think of sets of $(u-1)$ nodes as being analogous to single nodes in a normal graph with uniformity $2$. Furthermore, this suggests that, perhaps surprisingly, for $k\leq 2u-1$ the $k$-hypercycle problem requires the same running time as the  $k$-hyperclique problem in a $u$ uniform graph. 

\subsection{Algorithms for Minimum Hypercycle}
We will present two algorithms in this section. The first is the naive algorithm, presented for completeness. The second uses the definitions of 
\uCLR~and \uECLR{u}{k}~from the previous section (see Definitions \ref{def: u-CLR} and \ref{def: uk-ECLR}) and demonstrates fast algorithms for them. We can then use the \uECLR{u}{k}~algorithm to solve the $k$-hypercycle problem for $k\geq 2u-1$.

\subsubsection{Algorithm for k less than 2u-1}
This is the totally naive algorithm for the problem. We simply try all possible orders of nodes and check if they form a hypercycle.

\begin{lemma}
    The minimum $k$-hypercycle problem in a $u$-uniform hypergraph can be solved in time $O(k \cdot n^k)$, when $k$ is constant this is $O(n^k)$. 
    \label{lem:shortweightedkhypercyclealg}
\end{lemma}
\begin{proof}
Say we are given a hypergraph $G$ with vertex set $V$ and hyperedge set $E$.

For all choices of nodes $v_1,\ldots, v_k$ where $v_i \in V$ and $v_i \ne v_j$ if $i\ne j$ we check if they form a hypercycle and if it is a hypercycle we compute its weight. Specifically for all $i \in [1,k]$ we check if $(v_{i}, v_{(i+1 \mod{k})}, \ldots, v_{(i+k-1 \mod{k})}) \in E$ and if they all exist we sum the associated weights. We track the minimum weight we have seen so far.  This takes $O(k)$ time for each choice of nodes. There are $O(n^k)$ choices of nodes. 

This gives  a running time of $O(k n^k)$. Note that this is $O(n^k)$ if $k$ is constant. 
\end{proof}

\subsubsection{Algorithm for k at least 2u-1}
We start with the base case, which can be solved with the naive algorithm. We present an algorithm which gives a speedup if the graph is sparse\footnote{We will note without proof that a $n^{u-1} |E|$ algorithm exists for finding the minimum hypercycle if instead of keeping a $n^{2u-2}$ sized output matrix you replace that matrix with a hashtable and ask only for the minimum between those sets of nodes in $L$ and $R$ which have a path.}. 

\begin{lemma}
    The weighted \uCLR~problem has a $O(u \cdot n^{u-1}|E| + n^{2u-2})$ algorithm, where $|E|$ is the number of hyperedges in the graph. Note this is also $O(n^{2u-1})$.
    \label{lem:weighteduCLRAlg}
\end{lemma}
\begin{proof}
Recall the input is a graph $G=(V,E)$ which is a $n$-vertex $u$-uniform $(2u-1)$-circle-layered (weighted) hypergraph with vertex partitions $V_1, \cdots, V_{2u-1}$. 
Recall that $L=V_1 \times \cdots \times V_{u-1}$ and $R=V_{u+1} \times \cdots \times V_{2u-1}$ are sets of vertex partitions defined in the problem.

We first initialize the output matrix CLR$_{2u-1}(L,R)$ with all entries being $\infty$. This takes $O(n^{2u-2})$ time.  Now in $O(n^{u-1} \cdot |E|)$ time we will consider all tuples in $L$ combined with all hyperedges which span partitions $V_{u}, \ldots,V_{2u-1}$. A combination of $L$ and such a hyperedge specifies all nodes in a path, in $O(u)$ time check the length of the path. We then lookup the corresponding entry in CLR$_{2u-1}(L,R)$ and overwrite the existing entry if this new path length is shorter. 

Note that $|E| = O(n^u)$, there are only $O(n^u)$ possible sets of $u$ nodes. Thus the overall running time is $O(n^{2u-1})$.
\end{proof}

The key for this next lemma is that a minimum hyperpath of length $k$ can be computed quickly given minimum reachability information for hyperpaths of length $k-1$. Specifically, by tracking the $u-1$ nodes on the `end' of the $(k-1)$-hyperpath we can extend by one node as the only hyperedge the new node is involved in only concern the last $u-1$ nodes of the previous hyperpath. This lets us extend our path information. 
\begin{lemma}
    The weighted \uECLR{u}{k}~problem has a $O(n^{2u-2} + n^{u-1} \cdot |E|)$ algorithm. Note this running time is also $O(n^{2u-1})$.
    
    \label{lem:weightedukECLRAlg}
\end{lemma}
\begin{proof}
    First, note that given the information of the shortest hyperpaths from all tuples of nodes in $L=V_1 \times \cdots \times V_{u-1}$ to $R_{k-1}=V_{k-(u-1)} \times \cdots \times V_{k-1}$ the shortest hyperpath from $L$ to $R_k$ can be computed with just the hyperedges which span $V_{k-(u-1)}, \ldots, V_k$. There is optimal substructure, given that we are tracking the first and last $u-1$ nodes of the path. This is because the last new edge we want to add on to the path will depend on the last $u-1$ nodes. 

    Now, given this we want to compute CLR$_{k}$ from CLR$_{k-1}$ and the hyperedges which span $V_{k-(u-1)}, \ldots, V_k$ in $O(n^{u-1}\cdot|E| + n^{2u-2})$ time. We will call the set of hyperedges which span  $V_{k-(u-1)}, \ldots, V_k$ $E_{k}$ for convenience. 

    Initialize CLR$_{k}(L,R)$ in $O(n^{2u-2})$ time with every entry starting as $\infty$ distance. 
     For a given entry $(v_1,\ldots, v_{u-1}) \in L$ and a hyperedge $e = (v_{k-(u-1)},\ldots,v_{k}) \in E_k$ we will update CLR$_k$. For convenience we will name the left tuple $\ell = (v_1,\ldots, v_{u-1})$, and both right tuples $r_{k-1} = (v_{k-(u-1)},\ldots,v_{k-1})$ and  $r_{k} = (v_{k-(u-1)+1},\ldots,v_{k})$. We will let $w(e)$ be the weight of the edge we are considering. Given a previous value of an entry of CLR$_k$ and our new value computed from CLR$_{k-1}$ and an edge we update the new value to:
     $$CLR_{k}(\ell,r_k) \leftarrow \min \left(CLR_{k}(\ell,r_k), CLR_{k-1}(\ell,r_{k-1}) + w(e) \right).$$
     After performing these updates for every edge in $E_k$  $CLR_{k}$ will be populated with the lengths of the shortest paths from $L$ to $R_k$. 

     The total running time includes $O(n^{2u-2})$ for initializing the output and $O(n^{u-1} \cdot |E|)$ for updating for every $\ell \in L$ and $e \in E_k$. This gives a running time of $O(n^{2u-2} + n^{u-1} \cdot |E|)$. Because $|E| = o(n^u)$ we can bound this as $O(n^{2u-1})$.
\end{proof}

Now we will use these split up problems solve the hypercycle problem. Note that the problems as they have been split are basically a dynamic programming setup.
\begin{lemma}
    The minimum $k$-hypercycle problem in a $u$-uniform graph has a  $O(k^k \cdot u \cdot k \cdot (n^{u-1}|E| + n^{2u-2}))$ algorithm $k \geq 2u-1$ which succeeds with probability at least $2/3$. Note that $|E| = O(n^u)$ and thus the run time is also $O(k^k \cdot u \cdot k \cdot n^{2u-1})$
    \label{lem:longweightedkhypercyclealg}
\end{lemma}
\begin{proof}
    The first step we will need to do is color-coding. We are given a graph with no layers and will turn the graph into a layered graph by randomly partitioning. We will take the minimum cycle we return over all possible partitionings. 

    The procedure works as follows, we take each node in our graph and randomly assign it to one of the partitions $V_1, \ldots, V_k$. We only include hyperedges that have exactly one node from $V_i, \ldots, V_{i+u \mod u}$ (that is we only include the circle edges). Note that we are simply deleting hypercycles, any hypercycle that exists in this graph also existed in the original graph. Further note that if a particular hypercycle exists with nodes $a_1 \rightarrow a_2 \rightarrow \cdots \rightarrow a_k$ then if we happen to assign $a_i$ to $V_i$ the hypercycle will exist in our new graph. So, with probability at least $k^{-k}$ such a cycle exists. If we repeat this procedure $10 + 10 \cdot k^{-k}$ times the chance we successfully include the minimum cycle will be at least $2/3$.

    First we will show that we can compute CLR$_{k}(L,R)$ in $O(k \cdot u \cdot (n^{u-1}|E| + n^{2u-2}))$ time. To compute CLR$_{2u-1}(L,R)$ we simply use the algorithm from Lemma \ref{lem:weighteduCLRAlg} which takes $O(u \cdot  (n^{u-1}|E| + n^{2u-2}))$ time. From there given CLR$_{i}(L,R)$ we can compute CLR$_{i+1}(L,R)$ in time $O(u \cdot n^{2u-1})$ using the algorithm from Lemma \ref{lem:weightedukECLRAlg}. We need to run this extension procedure $k-2u-1$ times which is $O(k)$ giving us a bound of $O(k \cdot u \cdot (n^{u-1}|E| + n^{2u-2}))$ time. 

    Given CLR$_{k}(L,R)$ we will now find the minimum hypercycle. Note that the minimum hypercycle which starts on nodes $v_1,\ldots, v_{u-1}$ and ends on nodes $v_{k-u+2}, \ldots, v_k$ has a length of the minimum hyperpath which starts and ends on those same nodes plus the $u-1$ edges which span between partitions $V_{k-u+2}$ to $V_{u-1}$. Which means the minimum hypercycle starting on nodes $v_1,\ldots, v_{u-1}$ and ends on nodes $v_{k-u+2}, \ldots, v_k$ has length 
    $$CLR_{k}(L,R)[v_1,\ldots, v_{u-1},v_{k-u+2}, \ldots, v_k] + \sum_{i \in [0,u-2]} (v_{u-1-i}, \ldots, v_{1}, v_{k}, \ldots, v_{k-i}).$$
    Note that given CLR$_{k}(L,R)$ this computation takes time $O(n^{2u-2}\cdot u)$ to complete for all possible choices of nodes $v_1,\ldots, v_{u-1},v_{k-u+2}, \ldots, v_k$. We can track the minimum cycle over all of these options and return that value. 

    In total for each random color coding we take time $O(u \cdot  (n^{u-1}|E| + n^{2u-2}) + k \cdot u \cdot (n^{u-1}|E| + n^{2u-2}) + n^{2u-2}\cdot u)$ which is $O(k \cdot u \cdot (n^{u-1}|E| + n^{2u-2}))$ time. If we return the minimum cycle length we returned over all $10 + 10 \cdot k^k$ random color codings we will return the minimum cycle length in the original graph with probability at least $2/3$. This gives a total time of $O(k^k \cdot k \cdot u \cdot (n^{u-1}|E| + n^{2u-2}))$.
\end{proof}

\subsubsection{Putting Both Algorithms Together}
We can now do the very simple operation of using the faster algorithm in the appropriate situation. 

\begin{theorem}
    An algorithm for finding the minimum $k$-hypercycle in a $u$-uniform hypergraph with $n$ nodes with probability at least $2/3$ exists which runs in $O\left(\min(n^k, n^{2u-1})\right)$ time for constant $k$.
\end{theorem}
\begin{proof}
    For $k<2u-1$ we apply Lemma \ref{lem:shortweightedkhypercyclealg} to get a $O(n^k)$ algorithm. 

    Note that if $k$ is constant then so is $u$.
    For $k\geq 2u-1$ we apply Lemma \ref{lem:longweightedkhypercyclealg} getting a $O(n^{2u-1})$ algorithm. 
\end{proof}

\subsection{Lower-Bounds for Minimum Hypercycle}

\begin{corollary}
    Let $G$ be a $2$-uniform hypergraph on $n$ vertices $V$, partitioned into $k$ parts $V_1,\ldots,V_k$. 
    
    Let $\gamma_2(k)=k-\lceil k/2\rceil +1$.
    
    In $O(n^{\gamma_2(k)})$ time we can create a $\gamma_2(k)$-uniform hypergraph $G'$ on the same node set $V$ as $G$, so that $G'$ contains an $k$-hypercycle if and only if $G$ contains an $k$-hyperclique with one node from each $V_i$.
    
    If $G$ has weights on its hyperedges in the range $[-W,W]$, then one can also assign weights to the hyperedges of $G'$ so that a minimum weight $k$-hypercycle in $G'$ corresponds to a minimum weight $k$-hyperclique in $G$ and every edge in the hyperclique has weight between $[-\binom{\gamma_2(k)}{2}W,\binom{\gamma_2(k)}{2}W]$. Notably, $\binom{\gamma_2(k)}{2}\leq O(k^2)$.
\end{corollary}
\begin{proof}
Simply plugging in $u=2$ to Theorem \ref{thm:hclique_to_hcycle} (originally from \cite{LVW18}).
\end{proof}

\subsubsection{Odd Cycle Lengths}
We can now consider what value of $k$ corresponds to a given value of $u$. That is, if $u = k- \lceil k/2 \rceil +1$ given $u$ what is the $k$ we should consider? Consider $k=2\ell +1$ then note that $u = 2\ell+1 - \ell -1+1 = \ell +1$. So, for a given $u$ we should consider $\ell = u-1$ which corresponds to $k = 2u-2+1 =  2u-1$.

\begin{corollary}
    Let $G$ be a $2$-uniform hypergraph on $n$ vertices $V$, partitioned into $k = 2u-1$ parts $V_1,\ldots,V_{2u-1}$. Note that $u=k-\lceil k/2\rceil +1$.
    
    In $O(n^{u})$ time we can create a $u$-uniform hypergraph $G'$ on the same node set $V$ as $G$, so that $G'$ contains an $(2u-1)$-hypercycle if and only if $G$ contains an $k$-hyperclique with one node from each $V_i$.
    
    If $G$ has weights on its hyperedges in the range $[-W,W]$, then one can also assign weights to the hyperedges of $G'$ so that a minimum weight $k$-hypercycle in $G'$ corresponds to a minimum weight $(2u-1)$-hyperclique in $G$ and every edge in the hyperclique has weight between $[-\binom{u}{2}W,\binom{u}{2}W]$. Notably, $\binom{u}{2}\leq O(u^2)$.
    \label{cor:reverseUandKMin}
\end{corollary}
\begin{proof}
We simply plug in $k = 2u-1$ and this makes $\gamma_2(k)=u$.
\end{proof}

This gives an immediate implication 

\begin{corollary}
The minimum $(2u-1,u)$-hypercycle problem requires $n^{2u-1-o(1)}$ time if the minimum $(2u-1)$-clique hypothesis holds. 
\label{cor:min2u-1ishard}
\end{corollary}
\begin{proof}
    The minimum $(2u-1)$-clique hypothesis states that finding the weight of the minimum $(2u-1)$-clique requires $n^{2u-1-o(1)}$ time. Using the reduction from Corollary \ref{cor:reverseUandKMin} which takes $O(n^{u})$ time we can show that a $T(n)$ time algorithm for the minimum $(2u-1)$-hypercycle problem produces a $O(T(n) + n^u)$ time algorithm for minimum $(2u-1)$-clique. So the minimum $(2u-1)$-hypercycle problem requires $n^{2u-1-o(1)}$ time.
\end{proof}

\subsubsection{Even Cycle Lengths}

\begin{corollary}
    \label{cor: reverseUandKMineven}
    Let $G$ be a $2$-uniform hypergraph on $n$ vertices $V$, partitioned into $k = 2u-2$ parts $V_1,\ldots,V_{2u-2}$. Note that $u=2(u-1)-\lceil 2(u-1)/2\rceil +1 = u-1+1$.
    
    In $O(n^{u})$ time we can create a $u$-uniform hypergraph $G'$ on the same node set $V$ as $G$, so that $G'$ contains an $(2u-2)$-hypercycle if and only if $G$ contains an $k$-hyperclique with one node from each $V_i$.
    
    If $G$ has weights on its hyperedges in the range $[-W,W]$, then one can also assign weights to the hyperedges of $G'$ so that a minimum weight $k$-hypercycle in $G'$ corresponds to a minimum weight $(2u-1)$-hyperclique in $G$ and every edge in the hyperclique has weight between $[-\binom{u}{2}W,\binom{u}{2}W]$. Notably, $\binom{u}{2}\leq O(u^2)$.
\end{corollary}
\begin{proof}
We simply plug in $k = 2u-2$ and this makes $\gamma_2(k)=u$.
\end{proof}

This gives an immediate implication 

\begin{corollary}
The minimum $(2u-2)$-hypercycle problem requires $n^{2u-2-o(1)}$ time if the minimum $(2u-2)$-clique hypothesis holds in a graph with uniformity $u$. 
\label{cor:min2u-1ishardEven}
\end{corollary}
\begin{proof}
    The minimum $(2u-2)$-clique hypothesis states that finding the weight of the minimum $(2u-2)$-clique requires $n^{2u-2-o(1)}$ time. Using the reduction from Corollary \ref{cor: reverseUandKMineven} which takes $O(n^{u})$ time we can show that a $T(n)$ time algorithm for the minimum $(2u-2)$-hypercycle problem produces a $O(T(n) + n^u)$ time algorithm for minimum $(2u-2)$-clique. So the minimum $(2u-2)$-hypercycle problem requires $n^{2u-2-o(1)}$ time.
\end{proof}

\subsubsection{Combining to Produce Lower Bound}
To prove the same statement for shorter cycles we will use the following strategy. Consider, for example $k = 2u-3$. We can show hypercycles of this length are hard when the uniformity is $u-1$. We will show that if the problem is hard with smaller uniformity then it is hard for larger uniformity. Together this will let us show hardness for smaller $k$ for uniformity $u$.

The intuition for this next lemma is that a $u' > u$ hypercycle is more constraining than a $u$ hypercycle. Notably we have edges which are `longer'. So, we can simply include all possible extensions of edges to compute the original problem. We use color-coding to make our graph $k$-circle-layered to make the proof more straightforward. 

\begin{lemma}
Let $k^k = n^{o(1)}$.
If the $k$-hypercycle problem is $n^{k-o(1)}$ hard for some uniformity $u < k$ then the $k$-hypercycle problem is  $n^{k-o(1)}$ hard for all uniformities $u'$ where $2u' > k > u' \geq u$.
\label{lem:growUniformityFree}
\end{lemma}
\begin{proof}

We will take a given hypergraph, $G_o$, with uniformity $u$ and transform it into a circle layered graph $G$. We will do this process randomly where if no hypercycles exist in $G$ then none exist in $G_o$. Also, there is at least a $1/k^k$ probability that a hypercycle exists in $G$ if one existed in $G_o$. To do this we will use color-coding. We randomly assign each node in the graph to one of the partitions $V_1, \ldots, V_k$. Then we delete any edge that doesn't span from $V_i, V_{i+1 \mod k}, \ldots, V_{i+u-1 \mod k}$. As we are deleting edges no new hypercycles will exist. If a hypercycle existed in the original graph in the order $v_1, \ldots, v_k$ then it will exist in $G$ if $v_i$ is assigned to $V_i$ for all $i$ (any rotation of this assignment will also work) so the probability the cycle exists in $G_o$ is at least $1/k^k$. So, we can repeat this process $k^k \lg^2(n)$ times to solve the origional problem with high probability. Note that we can consider hypercycles which must respect the $k$-circle-layered graph structure, that is, must have exactly one node in each partition. 

The core idea is that we will take an edge $(v_i,...,v_{i+u-1}) \in G$ and create $n^{u'-u}$ edges from it one for all $(v_i,...,v_{i+u-1}, \ldots v_{i+u'-1})$ for all $v_j \in V_j$ where $j \in [i+u, i+u'-1]$ which we put in $G'$. This creates at most $n^{u'}$ edges which preserves hardness. We will show that using this method iff a $k$-cycle exists in $G$ then at least one hypercycle exists in $G'$. Note that because we are in a circle-layered hypergraph and $k$ is not a multiple of $u'$ for there to be a $k$-hypercycle the edges must include exactly one edge from $V_i$ to $V_{i+u'-1 \mod k}$ for each $i$. 

If $v_1, \ldots, v_k$ is a hypercycle in $G$ then there is a hypercycle in $G'$ which is formed by the $k$ hyperedges: 
$$(v_i, \ldots, v_{i+u'-1 \mod k}).$$
These will exist in $G'$ because the edges 
$$(v_i, \ldots, v_{i+u-1 \mod k})$$
exist in $G$ and we add all completions of the $u'-u$ nodes onto edges. 

If $v_1, \ldots, v_k$ is a hypercycle in $G'$ then the following hyperedges must exist in $G'$:
$$(v_i, \ldots, v_{i+u'-1 \mod k}).$$
Further, note that we only add such an edge spanning $V_i$ to $V_{i+u'-1 \mod k}$ when $(v_i, \ldots, v_{i+u-1 \mod k}) \in G$. 

So, a hypercycle in $G'$ corresponds to a single hypercycle in $G$ over the same nodes. Further note that a single hypercycle in $G$ has only one completion in $G'$ (because the nodes are fixed). Thus the count of the number of hypercycles in $G$ and $G'$ will be be the same. So, there is a one-to-one correspondence between hypercycles which respect the circle layered graph (have one node from each partition). 
\end{proof}

Now we can combine our hardness results for $2u-1$ and $2u-2$ hypercycle with the above lemma to show hardness for $k$-hypercycle when $k$ is small. 

\weightedCycleLB*
\begin{proof}
    From Corollary \ref{cor:min2u-1ishard} we have that minimum $(2\ell-1)$-hypercycle requires $n^{2\ell-1-o(1)}$ time in a graph with uniformity $\ell$ if the minimum $(2\ell-1)$-clique hypothesis holds. Then applying Lemma \ref{lem:growUniformityFree} we have that 
    if $2\ell-1 > u \geq \ell$ then $(2\ell-1)$-hypercycle requires $n^{2\ell-1-o(1)}$ time in a $u$ uniform graph. We can re-state this constraint in terms of $k = 2\ell-1$ as $k$-hypercycle requires $n^{k-o(1)}$ time if $k > u \geq (k+1)/2$. This corresponds to odd $k$ where $k \geq u+1$ and $2u-1 \geq k$ which gives us the range of odd $k$ in $[u+1, 2u-1]$.

    We can then apply Corollary \ref{cor:min2u-1ishardEven} which states that  minimum $(2\ell-2)$-hypercycle requires $n^{2\ell-2-o(1)}$ time in a graph with uniformity $\ell$ if the minimum $(2\ell-2)$-clique hypothesis holds. We can re-state this constraint in terms of $k = 2\ell-2$ as $k$-hypercycle requires $n^{k-o(1)}$ time if $k > u \geq (k+2)/2$. Which says that for even $k$ where $k \geq u+1$ and $2u-2 \geq k$ the problem is $n^{k-o(1)}$ hard. 

    Combining the even and odd cases we get that all for all $k$ where $k \in [u+1, 2u-1]$ the minimum $k$-hypercycle problem in a $u$-uniform graph requires $n^{k-o(1)}$ time. 
\end{proof}

\section{Worst-Case to Average-Case Reductions for Counting sub-Hypergraphs}
\label{sec:wc_to_ac}
\subsection{Overview of Approach}
In this section, we aim to give results on worst-case to average-case reductions for counting instances of a subhypergraph. We do this with an extension of the Inclusion-Edgesclusion technique (among others) introduced in \cite{factoredProblems}. In fact, we show here that the result in \cite{factoredProblems} can be extended to hypergraphs with some modifications to the approach. We begin with some definitions for the problems in question.

\begin{definition}\label{def: H}
The \textbf{counting }$\boldsymbol{H}$ \textbf{subgraphs in an }$\boldsymbol{H}$\textbf{-partite hypergraph (\#H)} problem takes as input a hypergraph $H$ and a $H$-partite $n$-node graph $G$ with the vertices partitioned into $k$ components, $V_1,...,V_k$, and asks for the count of the number of (non-induced) subgraphs of $G$ that have exactly one node from each of the $k$ partitions and contain the hypergraph $H$.
\end{definition}

\begin{definition}\label{def: UH}
The \textbf{uniform counting }$\boldsymbol{H}$ \textbf{subgraphs in an }$\boldsymbol{H}$\textbf{-partite hypergraph (U\#H)} problem takes as input a hypergraph $H$ and a $H$-partite $n$-node graph $G$ with the vertices partitioned into $k$ components, $V_1,...,V_k$, where every hyperedge between any partitions that have edges in $H$ is chosen to exist iid with probability $\mu$. The problem then asks for the count of the number of (non-induced) subgraphs of $G$ that have exactly one node from each of the $k$ partitions and contain the hypergraph $H$.
\end{definition}

Note that both problems only consider $G$ that are $H$-partite, and that U\#H is the uniform distribution over inputs to \#H.

\begin{definition}\label{def: HK}
The \textbf{counting }$\boldsymbol{H}$ \textbf{subgraphs in a }$k$\textbf{-partite hypergraph (\#HK)} problem takes as input a hypergraph $H$ and a $k$-partite $n$-node graph $G$ with the vertices partitioned into $k$ components, $V_1,...,V_k$, and asks for the count of the number of (non-induced) subgraphs of $G$ that have exactly $k$ nodes and contain hypergraph $H$, where each partition contains exactly one node.
\end{definition}

\begin{definition}\label{def: HER}
The \textbf{counting }$\boldsymbol{H}$ \textbf{subgraphs in a } \textbf\erdosRen \textbf{ hypergraph (\#HER)} problem takes as input a hypergraph $H$ and a \erdosRen hypergraph $G$ where every possible hyperedge exists with probability $\frac1b$, and asks for the count of the number of (non-induced) subgraphs of $G$ that have exactly $k$ nodes and contain hypergraph $H$.
\end{definition}

We now state the main result of Section \ref{sec:wc_to_ac}, the formalization of Theorem \ref{thm: informal wc to ac}.

\begin{restatable}{theorem}{WCtoACcounting}\label{thm: WCtoACcounting}
Let $H$ have $e$ edges and $k = O(1)$ edges, Let $A$ be an average-case algorithm for counting subgraphs $A$ in \erdosRen hypergraphs with edge probability $1/b$ which takes $T(n)$ time with probability at least $1-2^{-2^k} \cdot b^{-2^k} \cdot (\lg (e) \lg\lg(e))^{-\omega(1)}$.

Then there exists an algorithm $A'$ that runs in time $\Otil(T(n))$ that solves the \#HK problem with probability at least $1-\Otil(2^{\lg^2(n)})$.
\end{restatable}

Our goal is to do the following chain of reductions:
\[
\mathrlap{\underbrace{\phantom{\text{\#HK}\rightarrow\text{\#H}}}_\text{(1)}} \text{\#HK}\rightarrow\mathrlap{\overbrace{\phantom{\text{\#H}\rightarrow\text{U\#H}}}^\text{(2)}} \text{\#H}\rightarrow\underbrace{\text{U\#H} \rightarrow\text{\#HER}}_\text{(3)}
\]

Completing the $3$ reductions will show the desired result. We will do reduction $2$, then reduction $3$, then reduction $1$ in the following sections.

\subsection{Reducing \#H to U\#H}\label{sec: H to UH}

We begin by defining the notion of a \emph{good low-degree polynomial}, introduced by \cite{factoredProblems}.

\begin{definition}\cite{factoredProblems}\label{def: GLDP}
Let $n$ be the input size of a problem $P$, let $P$ return an integer in the range $[0,p-1]$ where $p$ is a prime and $p<n^c$ for some constant $c$. A good low-degree polynomial is a polynomial $f$ over a finite prime field $F_p$ where:

\begin{itemize}
    \item If $\Vec{I} = b_1,...,b_n$, then $f(b_1,...,b_n) = f(\Vec{I}) = P(\Vec{I})$, where $b_i$ is a zero or a one in the field.
    \item The function $f$ has degree $d = o(\lg(n)/\lg\lg(n))$.
    \item The function is strongly $d-$partite, meaning that the inputs can be partitioned into $d$ sets where no monomial contains more than one input that comes from the same set.
\end{itemize}
\end{definition}
\begin{definition}\label{def: GLDP H}
Let $H$ be a $k$-node graph with vertices $V_H$ and $G$ an $H$-partite $n$-node hypergraph with vertex set partition $V_1,...,V_k$. Let $E$ be the set of variables $\{e(v_{a_1},...,v_{a_c})|v_i\in V_i, a_x < a_y \iff x < y\}$ such that $e$ is 1 when the hyperedge between the vertices exists in $G$ and 0 when it does not. Let $h(v_{1},...,v_{k})$ be a function that multiplies all corresponding $e$ for hyperedges in $H$ for the selection of vertices in the input. Define $f$ as follows:
\[
f(E)=\sum_{v_1\in V_1,...v_k\in V_k} h(v_{1},...,v_{k})\mod{p}
\]

Where $p$ is some prime.
\end{definition}

\begin{lemma}\label{lem: f correct}
The function in the above definition returns the output of \#H, given that $p$ is a prime in $[2n^k, n^{2k}]$. 
\end{lemma}
\begin{proof}
Consider an arbitrary collection of inputs $v_1,...v_k$ into $h$. If this collection of vertices indeed contains $H$, then we see the output must be 1, as every corresponding edge variable in $H$ for that permutation of vertices must have value $1$. We then take the sum over all possible selections of vertices to arrive at our count.

Every single possible subgraph $H$ is found and counted in this way, as for every such subgraph, there must be a selection of vertices that contains it. Also, no subgraph is counted more than once, as only one selection of nodes can count any particular subgraph. Since the maximum number of subgraphs is upper bounded by $n^k$, the prime does not affect the correctness of the calculation.
\end{proof}

\begin{lemma}\label{lem: f GLDP}
$f$ is a good low-degree polynomial for \#H if the number of edges in $H$ is $o(\lg(n)/\lg\lg(n))$.
\end{lemma}
\begin{proof}
To prove the lemma, we note that $f$ is a polynomial over a prime finite field, and the number of monomials is $O(n^k\cdot k!)$, which is polynomial. By Lemma \ref{lem: f correct}, the function returns the same value as \#H.

Let $|E_H|$ be the number of edges in $H$. The function $f$ has degree $|E_H|=O (2^k - 1) $. This is because $h$ multiplies exactly as many variables together as the number of edges it has. We note that when $k$ is constant, the degree is constant.

Finally, $f$ is strongly $|E_H|$-partite. There are $|E_H|$ partitions of edges. $f$ is a sum over calls to $h$ where $h$ takes as input one variable from each edge partition and multiplies all of them.
\end{proof}

We can now apply the following result, also from \cite{factoredProblems}.

\begin{theorem}\cite{factoredProblems}\label{thm: GLDP}
Let $\mu$ be a constant such that $\mu\in(0,1)$. Let $P$ be a problem such that a function $f$ exists that is a good low-degree polynomial for $P$, and let $d$ be the degree of $f$. Let $A$ be an algorithm that runs in time $T(n)$ such that when $\Vec{I}$ is formed by $n$ bits each chosen iid from $Ber[\mu]$:

\[
Pr[A(\Vec{I}) = P(\Vec{I})]\geq 1-1/\omega\left(\lg^d(n)\lg\lg^d(n)\right).
\]

Then there is a randomized algorithm $B$ that runs in time $\Otil(n+T(n))$ such that for any $\Vec I\in\{0,1\}^n$:

\[
\Pr[B(\Vec I) = P(\Vec I)]\geq 1 - O\left(2^{-\lg^2(n)}\right).
\]
\end{theorem}

This immediately gives the following corollaries that finish the reduction.

\begin{corollary}
Let $d=2^k$ and $k=o(\sqrt[c]{\lg(n)/\lg\lg(n) })$. If an algorithm exists to solve U\#H in time $T(n)$ with probability $1-1/\omega(\lg^d(n)\lg\lg^d(n) )$, then an algorithm exists to solve \#H in time $\Otil(T(n) + n^2)$ with probability at least $1-O \left(2^ {-\lg^2(n)}\right)$.
\end{corollary}

\begin{corollary}\label{H TO UH}
Let $ H $ be such that $|E_H|=o(\lg(n)/\lg\lg(n) )$ and define $d = |E_H|$. If an algorithm exists to solve U\#H in time $T(n)$ with probability $1-1/\omega(\lg^d(n)\lg\lg^d(n) )$, then an algorithm exists to solve \#H in time $\Otil(T(n) + n^2)$ with probability at least $1-O \left(2^ {-\lg^2(n)}\right)$.
\end{corollary}

\subsection{Reducing U\#H to Average-Case Erd{\H{o}}s-R{\'{e}}nyi}

We now desire to reduce U\#H to Average-Case Erd{\H{o}}s-R{\'{e}}nyi - meaning that, we want to show that counting $H$ in an Erd{\H{o}}s-R{\'{e}}nyi hypergraph can be used to solve U\#H. We make a note here that, with more precise counting, the below techniques work for $k = o(\lg\lg n)$, but for ease of discussion, we will proceed assuming that $k$ is constant. We also note that we do not put any restrictions on the size of the hyperedges. 

\begin{definition}
Let $G$ be a $k$-partite Erd{\H{o}}s-R{\'{e}}nyi hypergraph with every hyperedge included with probability $1/b$, where $b$ is an integer. Let the vertex partitions of $G$ be $V_1,...,V_k$, and the edge partitions be $\edgepartitionc$, where $a_i < a_j \iff i < j$.

Label all hyperedges in $\edgepartitionc$ with $\ell\in[1,b]$ as follows: Edges that exist in $G$ are given label 1. The rest of the edges are uniformly assigned labels from $[2,b]$. Let $\edgepartitionc^\ell$ be the set of all edges of label $\ell$ between these $c$ vertex sets.

Let $\gchosenlabels$ be the following graph: We select all edges with label $\ell_i$ from the $i^{th}$ edge partition by lexicographical ordering on the labels of the vertex sets associated with the edge partition. We note that there are then $b^{2^k}$ such hypergraphs, as we have $b$ choices for each distinct edge partition. Define this set of hypergraphs as $S_G$. We note that, due to symmetry, all hypergraphs in $S_G$ are drawn from the same distribution.
\end{definition}

Essentially, we are looking at $G$ through the lens of a complete $k$-partite hypergraph with labels on the edges, labeling all the edges in $G$ with $1$. We then focus on hypergraphs where we choose a label for every edge partition, and take edges from each partition with the corresponding label.

\begin{definition}\label{def: count labels}
Let $G$ be defined as above. We define a labeled subgraph $L$ of $H$ in $G$ to be a subgraph of $H$ where every vertex is assigned a unique label in $[1,k]$. Define the count of the number of $L$ in $G$ to be the number of not-necessarily induced subgraphs $L$ where every vertex with label $\ell$ in $L$ comes from $V_\ell$ in $G$.
\end{definition}

Again, we want to reduce U\#H to counting subgraphs in Erd{\H{o}}s-R{\'{e}}nyi hypergraphs. The main barrier that we must overcome comes from the fact that, in an Erd{\H{o}}s-R{\'{e}}nyi graph, there exist edges that don't exist in $H$-partite graphs. This leads to overcounting subgraphs. We solve this problem by creating correlated hypergraphs 
that individually look Erd{\H{o}}s-R{\'{e}}nyi, through which we can get the true count of $H$ in the $H$ partite graph.

\subsubsection{Counting Small Subgraphs}
Our argument will make use of recursion to count labeled $H$. We begin here by solving the base cases. We define counting labeled $H$ in the same way as outlined in 
Definition \ref{def: count labels}.

\begin{lemma}\label{lem: disconnected count}
Let $G$ be a hypergraph with $n$ nodes, $m$ edges, and $k$ labeled partitions of the vertices, $V_1,...,V_k$ ($G$ is not necessarily $k$-partite).

If we have the counts of all labeled subgraphs of $H$ in $G$ of size less than $s$ vertices, we can compute the number of labeled subgraphs in $G$ that are the union of two disconnected labeled subgraphs of $H$ of size $s$ or less in constant time.
\end{lemma}

\begin{proof}
Let one be labeled subgraph $L$, and the other $L'$. Given that they share no vertices, we can simply multiply the numbers of the two subgraphs. This clearly takes constant time.
\end{proof}

\begin{lemma}\label{lem: tiny count}
Let $G$ be defined as above. We can compute the count of any subgraph $H$ in $G$ with 1 edge or fewer in $\Otil(m)$ time. This applies regardless of whether or not $H$ is labeled.
\end{lemma}

\begin{proof}
We can count these by iterating over all edges, and by counting the number of ways to select vertices using basic combinatorics. This takes $\Otil(km) =\Otil(m) $ time.
\end{proof}

\subsubsection{Recursion}

This step forms the main technical difficulty of this counting process. We will use all counts of subgraphs with a smaller number of hyperedges to count those with more hyperedges.

\begin{restatable}{lemma}{lemmaInclusion}\label{lem: inclusion}
Let $G$ be a labeled $k$-partite hypergraph with $n$ nodes per partition.

Say we are given the counts of the number of subgraphs $H$ in all hypergraphs in $S_G$.

Additionally, say we are given the counts of all labeled subgraphs of $H$ with $x\in [0,v]$ vertices and $y\in [0,e]$ hyperedges.

Let $L$ be a labeled subgraph of $H$ with $v$ vertices and $e+1$ edges.

Using all of these counts, we can count the number of not-necessarily induced subgraphs $L$ in $G$ in time $O(k!\cdot 2^{2^k} + b^{2^k})$.
\end{restatable}

The techniques we use in this proof essentially comes from a careful extension of the technique used to prove Lemma 5.9 in \cite{factoredProblems}. Of course, with there being hyperedges instead of edges, there are key parameters that change (resulting in a slightly different Lemma statement), but the key ideas are the same. For this reason, we leave the proof for this Lemma to Section \ref{sec: appendixLemmaRecursion}.

\subsubsection{Reducing to Erd{\H{o}}s-R{\'{e}}nyi}
We reduce counting labeled copies of $H$ in a $k$-partite Erd{\H{o}}s-R{\'{e}}nyi hypergraph to counting $H$ in Erd{\H{o}}s-R{\'{e}}nyi hypergraphs. Note that picking a particular labeling solves the problem of U\#H, as we can treat labeled hypergraph partitions as being $H$-partite.

\begin{lemma}\label{lem: AC to HP count}
Let $H$ have $e$ edges and $k$ vertices. Let $A$ be an average-case algorithm for counting unlabeled subgraphs $H$ in $k$-partite Erd{\H{o}}s-R{\'{e}}nyi graphs with edge probability $1/b$ which takes $T(n)$ time with probability $1 - \varepsilon/b^{2^k }$.

The number of labeled copies of subgraph $H$ in $k$-partite Erd{\H{o}}s-R{\'{e}}nyi graphs with edge probability $1/b$ can be computed in time $\Otil \left(2^{2^k} \cdot \left(m + k!\cdot 2^{2^k } + b^{2^k }\right) + b^{2^k  } \cdot T(n) \right) $ with probability $1-\varepsilon$.
\end{lemma}

We note that the technique we use here is slightly simpler than the one used to achieve the analogous statment in \cite{factoredProblems}.

\begin{proof}
Define $G$ to be a $k$-partite Erd{\H{o}}s-R{\'{e}}nyi hypergraph with edge probability $1/b$. Let $S_G$ and $\gchosenlabels$ be defined as above.

We aim to meet the conditions of Lemma \ref{lem: inclusion} for labeled subgraphs of size $v=k$ and $e=1$ . First, we require the counts of the number of subgraphs $H$ in $S_G$. We call $A$ on each one of the hypergraphs in $S_G$, again noting that each of them are drawn from the same distribution, including $G$. We make $b^{2^k -1 }$ such calls.

Next, we require the counts of all labeled subgraphs of $H$ with $x\in [0,k]$ vertices and $y\in [0,1]$ edges. By Lemma \ref{lem: tiny count}, we can do this in $\Otil( 2^{2^k} m)$ time.

By Lemma \ref{lem: inclusion}, we can now have the counts of any labeled subgraph of size $2$, using $O(k!\cdot 2^{2^k } + b^{2^k })$ time. We do this same process for all subgraphs of $H$, doing subgraphs of fewer edges first, to ensure we have all counts we need. Keep in mind that we do not need to call $A$ again, as the same counts still work.

In the end, we will have made $b^{2^k - 1 }$ calls to $A$. By union bound, this gives us $1-\varepsilon$ probability of failure.

We invoke Lemma \ref{lem: inclusion} once for every subgraph of $H$, of which there are $2^{2^k}$, giving us the desired runtime.
\end{proof}

\begin{lemma} \label{lem:WCtoAC}
Let $H$ have $e$ edges and $k$ vertices, and let $A$ be an average-case algorithm for counting subgraphs $H$ in Erd{\H{o}}s-R{\'{e}}nyi graphs where each hyperedge exists with probability $1/b$ which takes $T(n)$ time with probability $1-2^{-2k}\cdot b^{2^k}\cdot\left(\log(e)\log\log(e))^{-\omega(1)}\right)$.

Then, there exists an algorithm to count subgraphs in uniform $H$-partite graphs in time $\Otil(T(n))$ (U\#H) with probability at least $1-O(2^{-\log^2{n} } )$.

\end{lemma}

\begin{proof}
We begin with a note that, if we can count labeled subgraphs $H$ in Erd{\H{o}}s-R{\'{e}}nyi $k$-partite subgraphs, we can count subgraphs in uniform $H$-partite graphs with the same success probability. We simply add in the ``missing" hyperedges from the partitions not in $H$ in an Erd{\H{o}}s-R{\'{e}}nyi fashion, taking $\Otil(m)$ time.

Thus, by Lemma \ref{lem: AC to HP count}, it suffices to count $H$ in $k$-partite Erd{\H{o}}s-R{\'{e}}nyi graphs using $A$. Beginning with a $k$-partite Erd{\H{o}}s-R{\'{e}}nyi graph ($G$), we need to add random edges within each partition to make it look like a unpartitioned Erd{\H{o}}s-R{\'{e}}nyi graph. After adding the edges within the partitions with probability $1/b$ (call this resultant graph $G'$), we note that the count of $H$ in $G'$ that only uses one node from each partition in $G$ is still our answer. We can thus use inclusion exclusion to eliminate the counts of all instances of $H$ that do not use exactly one node from each partition in $G$. We call $A$ on every subset of the $k$-partitions, of which there are $2^k$, and we perform the inclusion-exclusion calculation. By union bound, the probability of this is at least $1-\left(\log (e) \log\log (e)\right)^{-\omega(1)} $.

\end{proof}

\subsection{\#HK to \#H and the proof of Theorem \ref{thm: WCtoACcounting}}

\begin{lemma}\label{lem: HK to H}

Let $A$ be an algorithm that solves \#H on $H$-partite $G$ with $n$ vertices and $H$ with $v = k = O(1)$ and $e = O(1)$ edges in time $T(n)$ with probability of success at least $1-\varepsilon$.

Then there exists an algorithm that solves \#HK on $k$-partite $G'$ with $n$ vertices and the same $H$ that runs in $O(T(n))$ time with probability of success at least $1 - 2^{2^k}\varepsilon$.

\end{lemma}
\begin{proof}
Intuitively, we are using the fact that every instance of $H$ that appears in $G'$ must have come from one labeling of the $k$ partitions that happens to be $H$-partite. Essentially, we are iterating over all possible labelings of the $k$-partitions, and adding up all of the counts. However, doing this naively can yield double-counting if there is the right symmetry in the structure of $H$ and the edges in $G'$. Thus, we iterate over all collections of $e$ of the $2^k$ edge partitions, using $A$ on the ones that are $H$-partite. There are at most $2^ {2 ^ k} $ such selections.

It's clear that each valid instance of $H$ in $G'$ will appear exactly once in one of the edge partition choices, namely, the one that contains each of the edges in that particular $H$. We run the algorithm once per valid edge partition, giving the desired runtime. Taking a union bound over the probability of failure, we get the desired probability of success.

\end{proof}

We can now prove the main theorem in this section.

\WCtoACcounting*

\begin{proof}
This immediately follows from applying Lemma \ref{lem:WCtoAC}, Corollary \ref{H TO UH}, and Lemma \ref{lem: HK to H}.
\end{proof}

\section{Applications to Database Problems}
\label{sec:database}
The aim of this section is to demonstrate the applicability of the above techniques to solve existing problems in the databases setting. We begin by defining the relevant terms.

A \emph{database} $D$ is a relational structure with a set of objects $V$, and relations $R_i(\Vec s)$, where $\Vec s\in V^r$. We refer to $r$ as the \emph{arity} of $R_i$. A relational term that exists in the database is referred to as a \emph{fact}.

A \emph{conjunctive query} $Q(R_a(\Vec s_a),...)$ is defined as a conjunction of the relations $R_j$ in its input. The variables in $X_i$ are referred to as \emph{free variables}. The query evaluates to true if there is some assignment of variables in $V$ to the free variables in the sets $X_i$ such that every relation in $Q$ exists in the database. We note that some of the sets $X_i$ may be share free variables. Naturally, every single free variable with the same label must be the same. In this paper, this is the only type of query we consider, and so we will simply refer to them as \emph{queries}.

A \emph{count query} simply answers with the number of unique assignments of variables that are valid for $Q$. A query contains \emph{self-joins} if it contains more than one copy of the same relation. We say that a query is \emph{self-join-free} if it has no self-joins. For this paper, we assume that the size of the queries is constant.

\subsection{Worst case to Average Case reduction for self-join-free count queries}

\begin{definition}
The \textbf{self-join-free count query (SCQ)} problem takes as input a database $D$ and a query $Q$ and asks for the count of the number of unique assignments of elements in the domain of $D$ to free variables in $Q$ such that $Q$ is satisfied. Relations in $D$ can be limited to those that appear in $Q$.
\end{definition}

\begin{definition}
The \textbf{uniform self-join-free count query (USCQ)} problem is the same as the SCQ problem, only that the database is such that every fact within a relation exists with some probability $\mu_{R_{i}}$. This is the average-case version of the SCQ problem.
\end{definition}

Much like the reduction we did in Section \ref{sec: H to UH}, we start by constructing a good low-degree polynomial for SCQ.

\begin{definition}
Let $D$ be a database with domain $V$ and relations $R_i$. Let $Q$ be a count query as defined above, where the number of relations in $Q$ is $k$ and the number of free variables is $r$. Note that $r\leq k$. Let $\Vec E$ be a vector such that every entry corresponds to a potential fact in the database, where it is $1$ if the fact is in the database, and $0$ otherwise. Let $h(v_1,...,v_r)$ be a function that multiples all corresponding entries for facts in the query for the selection of domain elements $v_1$ through $v_r$. Note that this function returns $1$ if the selection of variables in the input form a valid response to the query and $0$ otherwise. Define $f$ as follows:

\[
f(\Vec E) = \sum_{v_1,...,v_r\in V} h(v_1,...,v_r)\mod{p}
\]

where $p$ is some prime.
\end{definition}

\begin{lemma}\label{lem: SCQ polynomial accuracy}
The above function returns the output of SCQ given that $p$ is a prime in $[2n^k, n^{2k}] $
\end{lemma}

\begin{proof}
Consider an arbitrary term in the sum, defined by a selection of $r$ elements. As discussed in the definition of $f$, $h$ correctly identifies when the inputs form a valid response to the query. Also, every valid response to the query corresponds to one selection of free variable assignments, which is represented in only one term in the sum. Thus, acquiring the sum of all of these terms yields the correct count, before taking the modulo.

The number of possible correct queries is upperbounded by the total number of possible assignments to the free variables, which is $n^r\leq n^k$, so the modulo does not affect the count.
\end{proof}

\begin{lemma}
Let $ Q $ be such that $|Q|= O(1) $ and define $d = |Q|$.
$f$ is a good low-degree polynomial for SCQ of degree $d$ if the size of the query is bounded by a constant. \label{lem:scqGLDP}
\end{lemma}

\begin{proof}
We recall from Definition \ref{def: GLDP} that we need to show that $f$ has three properties.

First, by Lemma \ref{lem: SCQ polynomial accuracy}, we get the first property.

Second, the degree of the function is bounded by the degree of each of the terms. This is then the number of relations in the query, which is constant, as desired for the second property.

Lastly, the fact that each relation only appears once in the set (the query is self-join free) means that no monomial can ever appear more than once in a term, meaning that it is $d$-partite, where $d$ is the number of relations in the query.
\end{proof}

We now apply Theorem \ref{thm: GLDP} to get the result desired.

\begin{corollary}\label{cor: SCQ to USCQ}
Let $ Q $ be such that $|Q|= O(1) $ and define $d = |Q|$. Let the size of the query be constant. If an algorithm exists to solve USCQ in time $T(n)$ with probability $1-1/\omega(\lg^d(n)\lg\lg^d(n) )$, then an algorithm exists to solve SCQ in time $\Otil(T(n) + n^2)$ with probability at least $1-O \left(2^ {-\lg^2(n)}\right)$.
\end{corollary}
\begin{proof}
    We can apply Lemma \ref{lem:scqGLDP} to use $f$ as a GLDP of degree $d$ for SCQ. Then we can apply Theorem \ref{thm: GLDP} for the worst-case to average-case reduction given our GLDP for SCQ.
\end{proof}

\section{Conclusion and Open Problems}
\label{sec:conc_and_open_prob}
We have given a worst-case to average-case reduction for counting sub-hypergraphs and for counting database queries. We demonstrate the usefulness of our improved reduction by giving tight lower-bounds for average-case hypercycle counting for short hypercycles. 

Additionally, we present new algorithms and new lower bounds for hypercycle in the worst-case. We get tight bounds for the weighted problem of minimum-hypercycle. We get tight bounds for short hypercycle lengths for the unweighted problem. We show a new faster algorithm for $k$-hypercycle which depends on the running time of fast matrix multiplication. However, these results leave many problems open.

\subsection{Unweighted Hypercycle Open Problems}
Fundamentally the open problems here represent getting clear answers about the running time of the unweighted hypercycle problem when $k\geq \lambda_3^{-1}(u)+1$. Some specific open problems which we think would mark progress towards understanding this longer-cycle regime are:
\begin{enumerate}
    \item When $k = \lambda_3^{-1}(u)+1$ can you show a lower bound of $n^{\lambda_3^{-1}(u)-o(1)}$ for unweighted $k$-hypercycle? Note that if you could show this you would be giving a tight lower bound assuming $\omega=2$. 
    \item \textbf{(Extending the previous question)} When $ 2u-1> k \geq \lambda_3^{-1}(u)+1$ can you show a lower bound of $n^{k-1-o(1)}$ for unweighted $k$-hypercycle? 
    \item Can you give a reduction which shows that if (counting, directed) $k$-hypercycle in a $u$-uniform graph requires $T(n)$ time then so does (counting, directed) $(k+1)$-hypercycle? Note that we have a reduction which shows that if directed $k$-hypercycle in a $u$-uniform graph requires $T(n)$ time then so does directed $(k+u-1)$-hypercycle (directedness or something similar is needed for hypercycle lengths which are a multiple of $u$). This generalizes the idea of adding a partition to your graph with a matching which works for $u=2$. 
    \item \textbf{(Proving the previous point can't be done in general)} Can you show that for some $k$ and $u$ (counting, directed) $k$-hypercycle in a $u$-uniform graph requires $T(n)$ time and (counting, directed) $(k+1)$-hypercycle in a $u$-uniform graph can be solved in $o(T(n))$ time? For example, if you could show for some $u$ that $k = \lambda_3^{-1}(u)+1$ hypercycle could be solved in $n^{\lambda_3^{-1}(u)-\epsilon}$ for some $\epsilon >0$ this would be satisfied.
    \item Can you give a faster algorithm for the unweighted $k$-hypercycle problem than what we offer in this paper for any $k$ and $u$? 
\end{enumerate}

\subsection{Existential Hypercycle Open Problems}
There are interesting patterns we have noticed in $k$-hypercycles in $u$-uniform hypergraphs which  relate to the existence of graphs instead of the existence of algorithms. We note in the introduction that a theorem exists showing that $k$-hypercycles must exist in sufficiently dense graphs \cite{allen2015tight}. However, as that theorem is currently written it doesn't show that e.g. graphs with $n^{u-1}$ edges must have a $(2u)$-hypercycle. We note that when $k$ is not a multiple of $u$ you can create an extremely dense graph with no hypercycles. Specifically, create a $u$ partite with $n/u$ nodes in each partition and add every possible hyperedge which includes one node from each partition. This is metaphorically similar to no odd-cycles existing in a complete bi-partite graph. We will now give some concrete open problems related to this issue of multiple-of-$u$-hypercycles and density. 

\begin{enumerate}
    \item In $u=2$ uniform undirected graphs there is an interesting pattern which emerges. A fully dense bipartite graph contains no odd cycles, however, even cycles must exist even in relatively sparse graphs. Specifically, for constant $k$ in a graph with sparsity $\Omega(n^{1+1/k})$ there must exist a $2k$-cycle. How does this relate to hypergraphs of larger uniformity? Hypercycles of length $k = 0 \mod u$ can exist in a $u$-partite hypergraph, however hypercycles of length $k' \ne 0 \mod u$ do not exist in even a complete $u$-paritite hypergraph. This looks like it follows a metaphorically similar structure to graphs where $u=2$. Certainly it does for $k' \ne 0 \mod u$. However, we have not yet been able to characterize the density at which $k$-hypercycles are guaranteed to exist when $k = 0 \mod u$. A construction with $\Theta(n^u)$ edges where no $k$-hypercycle exists would settle the question, as would  a proof that a $k$-hypercycle must exist in any graph with $\Omega(n^{u-\epsilon})$ hyperedges for some constant $\epsilon >0$.
    \item To make the above question more concrete: what is the lowest hyperedge density such that a $6$-hypercycle is guaranteed to exist in any graph with that density in a $3$-uniform hypergraph? Is this density $\Theta(n^3)$? Is this density $\Omega(n^2)$?
    \item To make the above question (potentially) easier: You are given a tripartite 3-uniform hypergraph, $G$, with partitions $V_1, V_2, V_3$. Furthermore you are guaranteed that for every pair of nodes $(v_i, v_j) \in V_i \times V_j$ there are exactly $d$ edges of the form $(v_i, v_j, v_k)$ where $v_k \in V_k$ and $i\ne j \ne k$. This is a generalization of degree. What is the degree $d$ at which a $6$-hypercycle is guaranteed?
    \item Given a close reading of the proof of Theorem 1 from `Tight cycles in hypergraphs'(\cite{allen2015tight}) can you show that in a $u$-uniform hypergraph $G$ with at least $|E| \geq \frac{2k}{n} \binom{n}{u}$ hyperedges there must be a $k$-hypercycle? (This is what would happen if you could set $\delta = 2k/n$ and still have the proof go through.)
    \item For even cycles in 2-uniform graphs the longer the cycle the smaller the sparsity at which it is guaranteed to exist. Can something similar be proven for $u$-uniform graphs in general? 
\end{enumerate}

\subsection{Counting Color Coding}

We would love to say that counting constant sized subhypergraphs in the worst-case is equivalent in hardness up to log factors to counting constant sized subhypergraphs in the average-case. The issue we have here is, surprisingly, in the worst-case. We would like to show that counting a subhypergraph $H$ in a hypergraph is equivalent to counting $H$ in a $|H|$-partite graph. The issue is that standard approaches for this, even in $2$-uniform graphs, allow this reduction for \emph{detection} but not for counting. For specific graph structures (notably cliques) there are ways to build this reduction. However, for most graph structures getting the counts to line up seems untenable. Such a reduction would strengthen our results, but, also strengthen the many results which use color-coding a sub-routine. 

\let\realbibitem=\bibitem
\def\bibitem{\par \vspace{-0.5ex}\realbibitem}

\bibliographystyle{alpha}
\bibliography{my} 

\appendix



\section{Discussion}
\label{sec:appendix_discussion}
We start with a proof from the main body of the paper and then continue onto general discussion. 

\subsection{Proof of Recursion Lemma}
\label{sec: appendixLemmaRecursion}

We provide the proof of Lemma \ref{lem: inclusion}.

\lemmaInclusion*
\begin{proof}
Let $H$ have $k$ vertices, and $e_H$ hyperedges. Let $L$ be given as a list of vertices $v$ with labels corresponding to partitions $i_1,...i_v$, and $e+1$ hyperedges between partitions $i_{a_1},..., i_{a_c}$. Let $S_E$ be the set of all such sets $\{i_{a_1},...,i_{a_c}\}$. $\overline{S}_E$, then, will be the set of all sets of partitions not in $S_E$.

We bring our attention back to $S_G$. Consider the subset of instances in $S_G$ where the hyperedges between partitions in $S_E$ are labeled 1. These are essentially all graphs where the edges between partitions relevant to $L$ are in $G$. Call this subset $S_G[L]$.

Take the counts of the number of $H$ that appear in all graphs in $S_G[L]$, and call this count $c_{S_G[L]}$. This counts the number of $H$ that appear if the graph $G$ were to have all possible hyperedges between partitions in $\overline{S}_E$, weighted by how many edges in $S_E$ that subgraph uses. If $H$ appears in $G$ where $\ell$ of its edges are in the $\overline{S}_E$ partitions then it is counted $b^{2^k - 1 - e - 1 - \ell}$ times. 

Given that $L$ is a labeled subgraph of $H$, at least one labeling of $H$ will share all $e+1$ hyperedges and $v$ vertices of $L$. There may be many valid labelings for the hyperedges and vertices not in $L$.

We want to count all $H$ that have labelings that match the edges in $L$, and not those that only share some. Here is where we will use the additional counts of small graphs that we are given. For subgraphs that only partially match up with $L$, let's define their overlap to be $L'$. $L'$ must have $v$ vertices, and at most $e$ hyperedges. We want to somehow identify the counts of all subgraphs that overlap with any possible $L'$ without having edges in $L-L'$ and remove them from our previous count of $c_{S_G[L]}$.

Define $G_{L,L'}$ to be a hypergraph on $k$ vertices where all edges in $L-L'$ are excluded, and all other hyperedges are included. We define $c_{G_{L,L'}}$ to be the count of the number of subgraphs $H$ that exist in this graph that use all edges in $L'$. We can do this with brute force using $O(k!)$ time. As we will see later, we will need to do this computation on every possible subgraph of $L$, of which there are $O(2^{2^k })$.

Let $L'$ have $e_{L'}$ edges, and $c_{L'}$ be the count of labeled $L'$ that exist in $G$. The count of all subgraphs $H$ which overlap exactly $L'$ that are counted in $c_{S_G[L]}$ is:

\[
c_{L'} \cdot n^{k-v} \cdot c_{G_{L,L'}} \cdot b^{2^k - 1 - e - e_H + e_{L'}} 
\]

Let's try to break down this value. For every $\hat{L}$ that is counted by $c_{L'}$, we want to count the number of ways to construct $H$ around this. We first have to select the remaining vertices for our $H$. There are $n^{k-v}$ choices for this. Now that we have a particular selection of $k$ vertices, $c_{G_{L,L'}}$ tells us the number of ways that we can construct $H$ without using hyperedges in $L-L'$. For each one of these constructions of $H$, $\hat H $, we need to count how many times it appears in $c_{S_{G[L]}}$. $\hat H $ will only appear in the count in instances of the graph when the edges in $\hat H $ match the label that is selected for the partitions. However, for the edges that are not in $\hat H $, $\hat H $ will be counted regardless of the label selected. So $\hat H $ will be counted $ b^{ 2 ^ k -1 - e - e_H + e_{L'}}  $ times.

We note that this justification works any $L'\subseteq L$. Furthermore, since every possible $\hat H $ has to overlap with \textit{some} subgraph of $ L $ (keeping in mind that the empty graph is also a subgraph), we can account every single instance represented in $c_ { S_G [L] } $ in the following equation:

\[
\sum_{L'\subseteq L }  c_{L'} \cdot n^{k-v} \cdot c_{G_{L,L'}} \cdot b^{2^k-1 - e - e_H + e_{L'}}  = c_{S_G[L]}
\]

We now have every variable in this equation except for $c_L$, and we can now solve for it. This computation takes $O(k!\cdot 2^{2^k } + b^{2^k })$ to do, as desired.
\end{proof}

\subsection{Why Enumerating and Deciding  is Often Easy on Average}
\label{subsec:whyEnumIsEasy}
To build intuition we will start with easy graph cases. Then we will explain why enumerating or deciding small graph structures is always easy in random graphs. Then we will explain why a similar result holds true for databases. 

First, to build intuition, consider the case of triangle. Counting triangles is equivalently hard, up to log factors, in the worst-case and average-case \cite{UniformCliqueABB}. However, any set of $3$ nodes in an Erd{\H{o}}s-R{\'{e}}nyi graph have a $1/8$ chance of being a triangle. This makes enumerating and deciding if a clique exists easy. To enumerate we can start by taking linear time iterating through possible cliques to build up a backlog. Then after every $O(1)$ possible cliques that have been checked enumerate a new clique. With high probability there will be a saved clique to return. For deciding the existence of a clique one can simply return `yes' and be correct with probability at least $1 - (1-1/8)^{n/3}$.

\subsubsection{Hypergraphs}
The core ideas of both cases apply to all small graphs. If a subgraph (or subhypergraph) has a constant number of nodes $k$ then any one particular set of $k$ nodes will have a constant probability of being the relevant graph structure. There are at most $2^k$ edges and hyperedges in the subgraph and the chance that each of those edges exists or doesn't in the requested direction is $2^{-2^k}$. While this is an ugly looking relationship, it is constant if $k$ is constant (note if there are fewer edge constraints the relationship is much less negative). Now, given that the graph structure has a constant probability of existing we can say the probability of our subhypergraph existing in an Erd{\H{o}}s-R{\'{e}}nyi graph is at least 
$$1- (1 -2^{-2^k})^{n/k}.$$
If $k$ is constant this is $1 - 2^{-\Theta(n)}$. Enumerating these graph structures will similarly be easy with high probability. The expected number of these graph structures is high and enumerating possible subhypergraphs and returning when you find a subhypergraph which meets the condition is sufficient.

\subsubsection{Databases}
What about in databases? Well, similarly we have a case where we have a small structure we are looking for. In this case, rows that correspond to a query of interest. Every row exists iid with probability $1/2$, so once again, given a choice of elements the relevant queried rows exist with constant probability. This will, once again, make deciding very easy (simply return yes and you will be correct with high probability). When enumerating we can once again take the approach of taking linear time to build up a large number of instances and then after checking $O(1)$ possible query entries return an answer. If we check $\text{LARGE\_CONSTANT} \cdot 2^{2^|\text{SIZE\_OF\_QUERY}|}$ entries each time the probability that we ever run out of instances to return is low.

\end{document}